\documentclass[11pt]{article}

\usepackage[margin=1in]{geometry}
\usepackage{amsmath,amssymb,amsthm,mathtools}
\usepackage{booktabs}
\usepackage{enumitem}
\usepackage[protrusion=true,expansion=true,activate={true,nocompatibility},final,tracking=true,kerning=true,spacing=true,stretch=10,shrink=10,disable]{microtype}
\usepackage{natbib}
\usepackage[colorlinks=true,linkcolor=blue,citecolor=blue,urlcolor=blue]{hyperref}
\usepackage{setspace}
\usepackage{microtype}
\usepackage{pgfplots}
     \pgfplotsset{compat=1.17}

\newtheorem{assumption}{Assumption}
\newtheorem{proposition}{Proposition}

\newtheorem{corollary}{Corollary}
\newtheorem{remark}{Remark}
\newtheorem{definition}{Definition}

\newcommand{\ate}{\mathrm{ATE}}

\DeclareMathOperator{\Cov}{Cov}
\DeclareMathOperator{\Var}{Var}
\newcommand{\E}{\mathbb{E}}
\newcommand{\R}{\mathbb{R}}
\newcommand{\1}{\mathbf{1}}
\newcommand{\cZ}{\mathcal{Z}}
\newcommand{\cX}{\mathcal{X}}
\newcommand{\cG}{\mathcal{G}}
\newcommand{\cY}{\mathcal{Y}}
\newcommand{\dd}{\mathrm{d}}
\newcommand{\Perp}{\perp \!\!\! \perp}

\title{Stochastic Potential Choices and Outcomes\footnote{This paper is inspired by conversations with and notes by the late Gary Chamberlain.  We are grateful to Pedro Carneiro for comments on earlier drafts.  de Paula thanks the generous funding from the UK Research and Innovation (UKRI) under the UK government's Horizon Europe funding guarantee (Grant Ref: EP/X02931X/1) and the Economic and Social Research Council (ESRC) funding through the ESRC Institute for the Microeconomic Analysis of Public Policy (Grant Ref: ES/T014334/1).}}
\author{\'{A}ureo de Paula\thanks{Department of Economics, University College London, Institute for Fiscal Studies and Cemmap. } \and Elie Tamer\thanks{Department of Economics, Harvard University}}
\date{\today}

\begin{document}
\maketitle

\begin{abstract}
\noindent Applied econometricians typically model each individual as having fixed outcomes under treatment and control and, in instrumental-variables (IV) settings, fixed treatment decisions under each value of the instrument. This paper asks what changes when outcomes and treatment allocations or choices are stochastic at the individual level. In the model, each individual has a stable (but possibly stochastic) response type consisting of two objects: a treatment choice probability under each state, $p_i(z)$, and a potential outcome distribution under each treatment-state pair, $Q_i(.|d,z)$. These stochastic potential outcomes change the interpretation of some familiar estimators. For instance, in the deterministic IV model, the estimand identifies treatment effect only for compliers—those whose treatment status switches with the instrument. Under stochastic treatment allocation or choice there is no such discrete subgroup: the estimand averages effects over the population, weighting each individual by how much the instrument, policy, or assignment rule moves their probability of treatment.
The paper then gives an information-based foundation for stochastic choice, in which individuals act on expected gains given their information. Finally, repeated choices give the stable-kernel formulation empirical content: short panels identify moments of individual treatment probabilities, while long panels identify the joint dependence between treatment effects and the probability movements induced by an instrument or policy.
\end{abstract}


\newpage

\section{Introduction}

The standard potential outcomes framework in econometrics and statistics usually treats each individual's potential outcomes as fixed numbers (\cite{neyman1923}, \cite{rubin1974}, see also \cite{haavelmo1944}, p.23). In an instrumental-variables (IV) design, it also usually treats each individual's potential treatment choices as fixed numbers. For example, in a binary instrument with values $z'$ and $z$, individual $i$ has potential treatment statuses
\[
        D_i(z'),D_i(z)\in\{0,1\}.
\]
Changing the instrument value either changes the person's treatment status or it does not. Under monotonicity, the individuals whose status changes are the compliers, and the Wald estimand identifies the average treatment effect for that group (see \cite{imbensangrist1994},~\cite{angristimbensrubin1996}).

That representation is powerful, but has a flavor of determinism  (or, as put by \cite{dawid2000}, ``fatalism''), and can be too sharp for some empirical settings.  For instance, a reminder may make enrollment more likely without making it certain. A subsidy may raise take-up likelihood from $20\%$ to $50\%$ for one person and from $70\%$ to $75\%$ for another. A default rule may increase participation probabilities for some people while leaving others unchanged. A decision aid may draw some people into treatment and others out of treatment because it reduces mistakes on both sides of a threshold.   In fact, when asked about the chances of choosing particular occupations (\cite{arcidiaconoetal2020}), pregnancy and labour supply conditional on pregnancy status (\cite{briggsetal2024}) or buying a service under different scenarios (\cite{blasslachmanski2010}), respondents often provide probabilities strictly between $0\%$ and $100\%$.  In these examples, the stable object is not necessarily a deterministic treatment status under each instrument or policy value. It may instead be that person's non-trivial probability of treatment under each state.

This paper considers that possibility. It starts from stochastic potential outcomes and stochastic potential treatment choices or allocations. For each individual $i$ and state $z$, let
\[
        p_i(z)=\Pr(D_i(z)=1\mid G_i)
\]
denote the individual's probability of treatment under state $z$, conditional on a stable response type $G_i$. The type also contains a potential outcome distribution
\[
        Q_i(\cdot\mid d,z) = \Pr(Y_i(d,z)\in \cdot \mid G_i),
\]
so the potential outcome under treatment $d$ and state $z$ need not be a fixed number. The corresponding mean potential outcome is
\[
        \mu_i(d,z)=\int y\,Q_i(\dd y\mid d,z).
\]
We view $G_i \equiv (Q_i, p_i)$ as a person's stable ``type'' and explicitly condition the relevant probabilities above on $G_i$ to highlight this. The word stochastic refers to the fact that, even after conditioning on the type and fixing the state $z$, treatment receipt may remain random.  Similarly, the outcomes conditional on types for fixed state $z$ and treatment value $d$ may remain random.  In the present paper, the random type $G_i$ plays the role of the relevant background context (includes observed covariates), and $p_i(z')-p_i(z)$ is the type-specific probability change induced by a change in or contrast between states.

This approach has precedents in statistics and philosophy. In the statistics literature, stochastic counterfactual or potential outcomes generalize deterministic ones by allowing a given intervention to generate an individual-specific distribution of outcomes rather than a single fixed outcome \citep{greenland1987,robinsgreenland1989,robinsgreenland2000,vanderweelerobins2012}. Dawid's decision-theoretic approach similarly emphasizes regimes or decision settings and the conditions under which information can be transported across them \citep{dawid2000,dawid2021}. In philosophy, the probabilistic-causation literature starts from the idea that causes may raise or lower the probabilities of their effects rather than determine them \citep{hitchcock2023}. Cartwright’s work connects directly to the probabilistic-causation tradition, in which a cause is understood as raising the probability of its effects. Her central qualification is that probability raising is causally informative only relative to a suitably homogeneous background: unless other causally relevant factors are held fixed, observed probability differences may conceal, offset, or even reverse the operation of a genuine cause. Causal claims therefore cannot be reduced to probabilistic associations alone but require an account of the causal capacities or mechanisms that generate them \citep{cartwright1979,cartwright1989}. In this respect, the random-type formulation is close in spirit to probabilistic causation, but with an important qualification: the probability ($p_i(z)$) is not an observed association or a population-level probability-raising relation; it is a type-specific causal disposition, defined within a stable structural context. The state changes treatment probabilities because it activates or modifies a mechanism, not merely because it is statistically associated with treatment receipt. The framework has also been occasionally employed in econometrics (see, for example, \cite{abbringvandenberg2003} or \cite{dinardolee2011}).\footnote{See also related interesting work in \cite{chesher2021counterfactual}. Their analysis does not focus on stochastic potential outcomes, but on potential ``processes'' or ``sets of feasible counterfactual outcomes obtained by exogenously shifting'' the treatment or choice variable while holding what we call the state fixed. } 

For some parameters and estimation protocols, stochastic potential outcomes will lead to slight adaptations in interpretation.  For instance, an Average Treatment Effect (ATE) can be defined in terms of differences in potential \emph{mean} outcomes:
\[
        ATE = E[E\{Y_i(1)\mid G_i\} - E\{Y_i(0) \mid G_i\}] = E[\mu_i(1) - \mu_i(0)]
\]
where the unconditional expectations are taken across individuals and $\mu_i(d), d=0,1$ is the mean individual potential outcome under $d$ presented previously marginalized against the distribution of $Z$.  In this case, under appropriate conditions the usual estimators for the ATE under deterministic potential outcomes (e.g., difference in treatment and controls means when $\mu_i(d) \Perp D_i$) will aim at the ATE as generalized above for potential outcomes.\footnote{\cite{hernanrobins2025} point out that whether potential outcomes are stochastic or deterministic may have repercussions to uncertainty quantification depending on whether a model- or design-based inferential paradigm is adopted (see Section 10.3).} 

On the other hand, when the estimator is nonlinear in outcomes \emph{or} treatments (as is typical, for example, in IV methods), the same estimation protocols will correspond to estimands with a more distinct interpretation under stochastic potential outcomes and stochastic treatment choice or assignment. For example, in a deterministic model with a binary treatment and a binary instrumental variable, the instrumental variable partitions individuals into compliers, defiers, always-takers and never-takers (\cite{angristimbensrubin1996}). In the stochastic model, this is no longer the case: the instrument instead induces a distribution of probability changes across individuals. Those probability changes in turn determine weights in the estimand.

To elaborate on this, let $z$ and $z'$ denote two states-of-the-world. The states may correspond to an assignment rule, a price, a subsidy, a default rule, an information regime, or another policy environment. In a conventional IV application, the state is the instrumental variable; the outcome contrast given by $\E[Y_i\mid Z_i=z']-\E[Y_i\mid Z_i=z]$ is the reduced form; the corresponding treatment contrast given by $\E[D_i\mid Z_i=z']-\E[D_i\mid Z_i=z]$ is the first stage; and their ratio is the Wald estimand. We use the more general word state only because the same logic applies beyond binary instruments.

The key object is the individual responsiveness to a state contrast:
\[
        r_i(z',z)=p_i(z')-p_i(z).
\]
This is the stochastic analogue of a switching indicator. In the deterministic IV model, $D_i(z')-D_i(z)$ is equal to one for compliers, zero for never-takers and always-takers, and minus one for defiers.  Under monotonicity, the latter group is assumed away. In the stochastic model, responsiveness can be fractional. An individual with $p_i(z')=0.60$ and $p_i(z)=0.40$ contributes 0.20 units of responsiveness to the contrast $(z',z)$. She is not naturally an always-taker, never-taker, complier or defier. She is partially responsive.

The object of interest remains a contrast in potential outcomes. Under an exclusion restriction (i.e., $\mu_i(d,z) = \mu_i(d)$ for any $i, d$ and $z$), define the individual mean treatment effect as:
\[
        \Delta_i=\mu_i(1)-\mu_i(0).
\]
Variation in treatment helps us learn about these $\Delta_i$'s through the individual treatment probability movements it creates. Under independence of the state $Z_i$ from the response type $G_i$ and under the exclusion restriction above, the (average) outcome movement generated by the contrast $(z',z)$ is:
\[
        \E[Y_i\mid Z_i=z']-\E[Y_i\mid Z_i=z] = \E[\Delta_i r_i(z',z)],
\]
while the (average) treatment movement generated by the contrast $(z',z)$ is:
\[
        \E[D_i\mid Z_i=z']-\E[D_i\mid Z_i=z] = \E[r_i(z',z)].
\]
When the latter is nonzero, the outcome-treatment contrast ratio corresponds to the Wald estimand and identifies
\begin{equation}
        \theta(z',z)
        =\frac{\E[\Delta_i r_i(z',z)]}{\E[r_i(z',z)]}=\E[\Delta_i W_i(z',z)],
        \qquad
        W_i(z',z)=\frac{r_i(z',z)}{\E[r_i(z',z)]}.
        \label{eq:introtheta}
\end{equation}
Thus, the reduced form and the first stage reveal an average of treatment effects weighted by the pattern of treatment-probability movements generated by that particular source of variation (see, e.g., \cite{dinardolee2011}).

This formula nests familiar cases. If $r_i(z',z)$ is constant across individuals, the contrast recovers the ATE. If $p_i(z),p_i(z')\in\{0,1\}$ and $p_i(z')\geq p_i(z)$, treatment choice is deterministic and monotonic, and hence $r_i(z',z)\in\{0,1\}$.  Consequently, the contrast recovers the Local Average Treatment Effect (LATE) \citep{imbensangrist1994} for deterministic switchers. If $r_i(z',z)\geq 0$ but not necessarily zero or one, the contrast is a convex average weighted by probabilistic responsiveness. Perhaps more subtly, even if one interprets $r_i(z',z) \ge 0$ as a relaxation of the conventional monotonicity condition towards a stochastic or probabilistic monotonicity one (see, e.g., \cite{dinardolee2011}) and classifies those for whom $r_i(z',z) > 0$ as ``stochastic compliers,'' the estimand corresponds to a weighted average among those individuals with differing individual weights.  If $r_i(z',z)$ changes sign, the ratio is not a convex average. It is a signed comparison between treatment effects among types whose treatment probabilities rise and treatment effects among types whose treatment probabilities fall.

The distinction we make between an ex ante stochastic representation and an ex post deterministic representation is not a claim about different sampling protocols for the observed data. As Appendix~\ref{app:equivalence} shows, in a static cross-section the two representations generate the same observable conditional means relevant for reduced forms, first stages and Wald estimands. In that environment, the value of the stochastic formulation is interpretive: it takes \(z\mapsto p_i(z)\) as the stable ex ante object and uses changes in this map to reinterpret the treatment-effect parameters targeted by familiar econometric estimators. With repeated choices, however, stability of this map imposes additional restrictions and gives the formulation empirical content. Hence, the main message is simple. Under stochastic potential outcomes and stochastic choice, an instrument or policy does not merely select compliers. It induces a weighting operator over treatment effects. To interpret a reduced form, first stage, or Wald ratio, one must ask whose treatment probabilities are changed, by how much, and in which direction.

To that end, the paper makes four contributions. First, it revisits the stochastic potential outcomes model and views it in terms of compliance and treatment response. Each individual has a treatment choice probability under each state and a potential outcome distribution under each treatment-state pair. Deterministic potential outcomes and deterministic compliance types are nested as degenerate cases. This modeling change does not by itself create new estimators; it changes the interpretation of familiar estimators.

Second, the paper characterizes what standard contrast estimators target under this stochastic formulation. The IV/Wald estimand is not, in general, the mean effect for a discrete complier group. It is a weighted average over individuals, with weights proportional to how much the instrument changes each person's probability of treatment. In the deterministic LATE model, only compliers get positive and equal weight. In the stochastic model, all individuals may enter the estimand, but with different weights. This also clarifies when the estimand is an average and when it is only a signed contrast.

Third, the paper gives a possible behavioral interpretation of the weights. In a threshold-choice model in which the state changes either the perceived treatment margin or the information used to form it, local responsiveness decomposes into a margin-density term and a mechanism-loading term. The margin-density term captures how close the person is to indifference; the mechanism-loading term captures how strongly the instrument or policy moves that person's adoption margin. This generalizes the familiar IV insight. With deterministic compliance,
different instruments may identify different effects because they move different
complier groups. With stochastic choice, there need not be discrete complier
groups; instruments instead differ in the treatment-probability movements they
induce across individuals.

Fourth, the paper studies repeated observations on the same individuals. Within a stable, conditionally independent choice-kernel model, short panels identify moments of the individual treatment probabilities and give one-sided responsiveness refutable implications. Long panels identify the joint schedule relating mean treatment effects to the vector of treatment probabilities across states. This schedule reveals the latent weighting operator behind IV and permits researchers to compare the estimands generated by different instruments and first-stage specifications. The deterministic, time-invariant response model is a degenerate boundary of this panel model; the paper is explicit that the resulting tests concern this boundary, rather than determinism without restrictions on latent choice paths.

\section{Stochastic Potential Treatments and Outcomes}
\label{sec:model}

For each individual $i$, we observe
\[
        O_i=(Y_i,D_i,Z_i,X_i),
\]
where $Y_i\in\R$ is the observed outcome, $D_i\in\{0,1\}$ is the observed treatment receipt, $Z_i\in\cZ$ is the observed state, and $X_i\in\cX$ is a vector of covariates.

The primitive object is a stable (stochastic) response type. For each individual $i$, let $G_i\in\cG$ denote a latent type containing the following two kernels.  

\begin{definition}[Stochastic Response Type] \label{def:sttypes}
A stochastic response type $G_i$ consists of:
\begin{enumerate}[label=(\roman*)]
    \item a treatment choice kernel
    \[
        p_i(d\mid z)=\Pr(D_i(z)=d\mid G_i),
        \qquad d\in\{0,1\},\ z\in\cZ,
    \]
    with $p_i(z)=p_i(1\mid z)$, treatment values $d \in \{0, 1\}$ and state values $z \in \cZ$; and
    \item an outcome kernel
    \[
        Q_i(A\mid d,z)=\Pr(Y_i(d,z)\in A\mid G_i),
    \]
    for measurable sets $A\subseteq\R$, treatment values $d\in\{0,1\}$, and state values $z\in\cZ$.
\end{enumerate}
In addition, the type includes covariates $X_i$: $G_i \equiv (X_i, Q_i, p_i).$
\end{definition}

Note that the covariates $X_i$ are included in the definition of the type $G_i$.
Also, the notation in the above definition is meant to indicate that this outcome kernel is an autonomous  ``production function'' at the individual level and can be evaluated at any value of $(d,z)$ (i.e., it is a ``causal object''). This means that the distribution $Q_i(\cdot\mid d,z)$ is not the observed conditional distribution of $Y_i$ given $D_i=d$ and $Z_i=z$. It is the type-specific distribution of the potential outcome under treatment value $d$ and state value $z$.  The corresponding mean potential outcome under treatment value $d$ and state value $z$ is thus:
\[
        \mu_i(d,z) \equiv \int y\,Q_i(\dd y\mid d,z).
\]
Similarly, the distribution $p_i(\cdot \mid z)$ is not the observed conditional distribution of $D_i$ given $Z_i=z$.  It is the type-specific distribution of the potential treatment under state value $z$.

The observed variables are such that:
\[
        D_i=D_i(Z_i)
        \qquad \textrm{and} \qquad
        Y_i=Y_i(D_i(Z_i),Z_i).
\]

When potential outcomes and treatments are deterministic, the objects
\(D_i(\cdot)\) and \(Y_i(\cdot)\) are individual-specific functions. In the
stochastic formulation, however, even after fixing the type \(G_i\) and the
state \(z\), treatment receipt may remain random.  Similarly, fixing $G_i$, $d$ and $z$, outcomes may remain random.

The type \(G_i\) has been defined through two marginal objects: the
treatment choice kernel \(p_i(\cdot\mid z)\) and the potential outcome kernel
\(Q_i(\cdot\mid d,z)\). These marginal kernels are the objects we want to use
in the estimand. But by themselves they do not specify how the residual
randomness in treatment choice is related to the residual randomness in
potential outcomes within the same type. We therefore impose the following
mean-separability condition.

\begin{assumption}[No within-type selection on residual outcome shocks]
\label{ass:separability}
For every state value \(z\in\mathcal Z\) and treatment value
\(d\in\{0,1\}\),
\[
\E\{Y_i(d,z)\mid G_i,D_i(z)=d\}
=
\E\{Y_i(d,z)\mid G_i\}
=
\mu_i(d,z),
\]
whenever \(\Pr(D_i(z)=d\mid G_i)>0\).
\end{assumption}

Assumption~\ref{ass:separability} is not an independence condition for the
state or instrument. It is a condition on the stochastic representation within
a fixed type. It says that, after conditioning on \(G_i\), the residual event
that treatment is realized under state \(z\) does not change the mean of the
potential outcome under \((d,z)\).\footnote{A sufficient condition is a residual-shock representation
\(D_i(z)=D(z,G_i,U_i^D)\), \(Y_i(d,z)=Y(d,z,G_i,U_i^Y)\), with
\(U_i^Y\Perp U_i^D\mid G_i\). This is stronger than needed: the assumption
requires only mean independence of the residual outcome shock from the residual
treatment choice event within type. \label{assumption1}} The assumption allows arbitrary selection
across types. For example, types with larger treatment gains may also have
larger treatment probabilities. What it rules out is residual mean selection
within the same type.

Under Assumption~\ref{ass:separability}, the type-specific mean outcome under
state \(z\) can be written as the simple mixture
\[
m_i(z)
=
\sum_{d\in\{0,1\}}p_i(d\mid z)\mu_i(d,z)
=
\{1-p_i(z)\}\mu_i(0,z)+p_i(z)\mu_i(1,z).
\]
Thus \(m_i(z)\) is the mean outcome for type \(i\) when the state is set to
\(z\), averaging over the stochastic treatment choice generated by that state. The role of Assumption~\ref{ass:separability} is only to justify this mixture
representation using the marginal kernels \(p_i(\cdot\mid z)\) and
\(Q_i(\cdot\mid d,z)\).\footnote{If one instead defined the type more broadly as the
full joint law of
$\{D_i(z),Y_i(0,z),Y_i(1,z)\}_{z\in\mathcal Z},$ then one could take
$m_i(z)=\E\{Y_i(D_i(z),z)\mid G_i\}$
as primitive, and the no-within selection assumption would not be needed.
Similarly, in the deterministic ex post formulation, the condition is automatic:
once the full type is fixed, there is no residual treatment choice or outcome
randomness left.}

\section{What Observed Contrasts Identify}
\label{sec:observedcontrasts}

For two state values $z,z'\in\cZ$, define the \emph{observed} outcome and treatment contrasts as:
\[
        C_Y(z',z) \equiv \E[Y_i\mid Z_i=z']-\E[Y_i\mid Z_i=z] ~ \textrm{and} ~ C_D(z',z) \equiv \E[D_i\mid Z_i=z']-\E[D_i\mid Z_i=z].
\]
Without assumptions, these are descriptive contrasts. To give them causal content, the state must be independent of the latent response type, either unconditionally or conditional on covariates. 

First, we nonetheless start with a simple observation regarding the worst case bounds for the mean treatment effect. Then we state the independence assumption and derive the causal interpretation of the contrasts under that assumption.

\begin{remark}[Manski Worst-Case Bounds as Bounds Over Random Types]
\label{rem:manski}
Suppress the state argument or, equivalently, consider a fixed state and suppose that
\(Y_i \in [\underline y,\overline y]\) almost surely. Let
\[
\overline{p}=\Pr(D_i=1), \qquad
\overline{m}_1=E[Y_i\mid D_i=1], ~ \textrm{and} ~
\overline{m}_0=E[Y_i\mid D_i=0],
\]
where we use overline notation to distinguish the observable moments above from the latent ones introduced previously.  Under consistency and Assumption \ref{ass:separability}, the observed treated outcomes identify
the treated-population component
\[
E[D_i\mu_i(1)] = \overline{p} \times \overline{m}_1,
\]
and the observed untreated outcomes identify the untreated-population component
\[
E[(1-D_i)\mu_i(0)] = (1-\overline{p}) \times \overline{m}_0.
\]
The missing components (i.e., mean outcome under treatment for those untreated and mean outcome without treatment for those treated) are bounded only by the outcome support:
\[
(1-\overline{p})\underline y
\le
E[(1-D_i)\mu_i(1)]
\le
(1-\overline{p})\overline y,
\]
and
\[
\overline{p}\underline y
\le
E[D_i\mu_i(0)]
\le
\overline{p}\overline y.
\]
Therefore, the sharp worst-case bounds for
\[
ATE = E[\mu_i(1)-\mu_i(0)]
\]
are
\[
\overline{p}\,\overline{m}_1
+(1-\overline{p})\underline y
-\{\overline{p}\,\overline y+(1-\overline{p})\overline{m}_0\}
\leq
ATE
\leq
\overline{p}\,\overline{m}_1
+(1-\overline{p})\overline y
-\{\overline{p}\,\underline y+(1-\overline{p})\overline{m}_0\}.
\]

\noindent The stochastic potential outcomes formulation does not sharpen these Manski-style bounds (\cite{Manski2003,Manski1990}).
It only changes their interpretation. The identified set ranges over all random-type
distributions and outcome distribution  \(Q_i(\cdot\mid d)\) consistent with the observed
treatment-specific outcome means and the support restriction. Treated observations identify
the treated-population mean component of the treated potential outcome distribution, not each
treated individual's latent mean \(\mu_i(1)\). Analogously, untreated observations identify the
untreated-population mean component of \(\mu_i(0)\), not each untreated individual's latent
mean. Thus, for mean effects, completing missing outcome kernels yields the same worst-case
identified set as completing missing deterministic potential outcome values. No state-independence or
treatment choice independence assumption is used.
\end{remark}

 The above bounds obviously hold without any assumptions (other than support conditions and Assumption \ref{ass:separability}). Now, we state the independence assumption that is key to obtaining causal interpretations of the contrasts introduced earlier.
\begin{assumption}[State Independence]
\label{ass:stateindependence}
For the states being compared, either
\[
        Z_i\Perp G_i,
\]
or
\[
        Z_i\Perp G_i\mid X_i,
\]
with overlap for the relevant states at each covariate value in the target population. In addition, conditional on \(G_i\), assignment is independent of the residual draws realizing \(D_i(z)\) and \(Y_i(d,z)\), for all relevant \(d\) and \(z\).
\end{assumption}
This is the analogue of random assignment or instrument independence.   We also require that, conditional on $G_i$ and, where applicable, $X_i$, assignment is independent of the residual draws that realize the treatment choice and potential outcome kernels. In potential outcome terminology, it implies exchangeability of the relevant potential state outcomes across values of $Z_i$. (Under deterministic potential variables, this reduces to independence from
$\{D_i(z),Y_i(0,z),Y_i(1,z):z\in\mathcal Z\}$.)

The primitive causal object associated with changing the state from $z$ to $z'$ is the average state effect
\begin{equation}
        \Psi(z',z)=\E[m_i(z')-m_i(z)].
        \label{eq:stateeffect}
\end{equation}
This is the average effect of assigning state $z'$ rather than state $z$. At this level, the state may affect outcomes by changing treatment probabilities, by changing outcome distributions directly, or both. The next result shows how this average state effect relates to observables. 

Under unconditional state independence,
\begin{align*}
        \E[Y_i\mid Z_i=z]
        &=\E\{\E[Y_i\mid Z_i=z,G_i]\mid Z_i=z\} \\
        &=\E[m_i(z)\mid Z_i=z] \\
        &=\E[m_i(z)].
\end{align*}
The first equality is the law of iterated expectations. The second uses the definition of the type-specific mean outcome under state $z$. The third uses $Z_i\Perp G_i$, since $m_i(z)$ is a function of $G_i$. Thus
\[
        C_Y(z',z)=\Psi(z',z).
\]
No exclusion restriction is needed for this interpretation. A reduced form can be a causal effect of the state even when it is not a treatment effect.

\begin{proposition}[Reduced form as the causal effect of the state]
\label{prop:state_itt}
Suppose Assumptions~\ref{ass:separability} and the unconditional version of Assumption \ref{ass:stateindependence} hold.
Then, for each state value \(z\),
\[
        \E[Y_i\mid Z_i=z]=\E[m_i(z)].
\]
Consequently, for any two state values \(z,z'\in\mathcal Z\),
\[
        C_Y(z',z)
        =
        \E[Y_i\mid Z_i=z']-\E[Y_i\mid Z_i=z]
        =
        \E[m_i(z')-m_i(z)]
        =
        \Psi(z',z).
\]
Thus the reduced form identifies the causal effect of assigning state \(z'\)
rather than state \(z\). No exclusion restriction is required for this
interpretation.
\end{proposition}

\begin{proof}
By iterated expectations,
\[
\E[Y_i\mid Z_i=z]
=
\E\{\E[Y_i\mid Z_i=z,G_i]\mid Z_i=z\}.
\]
By the definition of the type-specific mean outcome under state \(z\),
\[
\E[Y_i\mid Z_i=z,G_i]=m_i(z).
\]
Therefore,
\[
\E[Y_i\mid Z_i=z]
=
\E[m_i(z)\mid Z_i=z].
\]
Since \(m_i(z)\) is a function of \(G_i\) and \(Z_i\Perp G_i\),
\[
\E[m_i(z)\mid Z_i=z]=\E[m_i(z)].
\]
Taking the difference between \(z'\) and \(z\) gives the result.
\end{proof}
The proposition says that the reduced form, or outcome intention-to-treat (ITT), is already a causal
effect of the assigned state. If \(Z_i\) is an encouragement, information
treatment, default rule, subsidy offer, or assignment rule, then
\(C_Y(z',z)\) is the causal effect of that intervention. The exclusion
restriction is needed only for the additional IV interpretation that this
effect operates solely through treatment receipt \(D_i\).

\begin{remark}[Covariate adjustment]
The unconditional independence assumption can be replaced by
\[
Z_i\Perp G_i\mid X_i.
\]
Then the same formulas hold after replacing raw contrasts by covariate-adjusted
contrasts. For example, one may use
\[
\int
\left\{
\E[Y_i\mid Z_i=z',X_i=x]
-
\E[Y_i\mid Z_i=z,X_i=x]
\right\}
\,dF_X(x),
\]
where \(F_X\) is the covariate distribution in the population of interest.
\end{remark}

\section{Interpreting Familiar Estimators}
\label{sec:familiar}

The previous section showed that, under state independence, the observed
outcome contrast
\[
        C_Y(z',z)
        =
        \E[Y_i\mid Z_i=z']-\E[Y_i\mid Z_i=z]
\]
identifies the causal ITT effect of assigning state \(z'\) rather than state
\(z\):
\[
        C_Y(z',z)
        =
        \E[m_i(z')-m_i(z)].
\]
This is already a causal effect of the assigned state \(Z_i\). For example,
if \(Z_i\) is an encouragement, default rule, information treatment, or subsidy
offer, then \(C_Y(z',z)\) is the ITT effect of that intervention.

One is perhaps more often interested in leveraging variation in the state to retrieve the direct causal effect of treatment receipt \(D_i\) on the outcome.  Aside from the assumptions introduced previously, this requires
an exclusion restriction. 

\begin{assumption}[Exclusion Restriction]
\label{exclusion}
Holding treatment fixed, the state does not directly
affect the outcome:
\[
        Q_i(\cdot\mid d,z)=Q_i(\cdot\mid d)
        \qquad d\in\{0,1\}, \forall z \in \cZ.
\]
\end{assumption}
This implies that
\[
        \mu_i(d,z)=\mu_i(d)
        \qquad d\in\{0,1\}.
\]
\noindent Under exclusion, define then the individual mean effect of treatment receipt by
\[
        \Delta_i=\mu_i(1)-\mu_i(0).
\]
Since
\[
        m_i(z)
        =
        \{1-p_i(z)\}\mu_i(0)+p_i(z)\mu_i(1),
\]
we can write
\[
        m_i(z)
        =
        \mu_i(0)+p_i(z)\Delta_i.
\]
Therefore,
\[
        m_i(z')-m_i(z)
        =
        \Delta_i\{p_i(z')-p_i(z)\}.
\]

It follows that, under no within-sample selection, state independence and exclusion, the outcome ITT or reduced form is
\[
        C_Y(z',z)
        =
        \E[\Delta_i\{p_i(z')-p_i(z)\}],
\]
while we can similarly obtain that the treatment ITT, or first stage, is
\[
        C_D(z',z)
        =
        \E[D_i\mid Z_i=z']-\E[D_i\mid Z_i=z]
        =
        \E[p_i(z')-p_i(z)].
\]
When the first stage is nonzero, the familiar outcome-treatment contrast ratio corresponds to the usual Wald-type object.  The numerator
is the ITT effect of the assigned state, written as a treatment-channel effect.
The denominator is the first stage. Given the conditions above, it identifies
\begin{equation} \label{eq:contrast_ratio}
        \frac{C_Y(z',z)}{C_D(z',z)}
        =
        \frac{\E[\Delta_i\{p_i(z')-p_i(z)\}]}
        {\E[p_i(z')-p_i(z)]}.
\end{equation}

This is the basic formula behind the examples in this section. The stochastic-choice interpretation is that this ratio averages individual
mean treatment effects using weights proportional to how much the state changes
each individual's probability of receiving treatment.
\subsection{Randomized assignment with perfect and imperfect compliance}

Suppose $Z_i\in\{0,1\}$ is randomly assigned. This is a setting of a randomized experiment and hence unconditional state independence holds. If assignment mechanically determines treatment receipt, then
\[
        p_i(1)=1,
        \qquad
        p_i(0)=0
\]
for all individuals. Hence $r_i(1,0)=1$, and under exclusion
\[
        \E[Y_i\mid Z_i=1]-\E[Y_i\mid Z_i=0]=\E[\Delta_i]=\ate.
\]
As is well-known, the perfect-compliance randomized experiment identifies the ATE because the assignment changes every person's treatment probability by the same amount.

With \emph{imperfect} individual compliance (i.e., stochastic potential outcomes and treatments), assignment changes treatment probabilities but need not determine treatment status for a given individual. Under no-within selection, state independence and exclusion, the contrast ratio is given by equation (\ref{eq:contrast_ratio}) and therefore identifies the assignment-responsiveness-weighted average of treatment effects. Because the weights $W_i(1,0)=r_i(1,0)/E[r_i(1,0)]$ average to one, this
weighted average decomposes as
\[
\frac{C_Y(1,0)}{C_D(1,0)}
= E[\Delta_i] + \Cov\bigl(\Delta_i, W_i(1,0)\bigr)
= ATE + \frac{\Cov\bigl(\Delta_i,\, p_i(1)-p_i(0)\bigr)}{E[p_i(1)-p_i(0)]}.
\]
The estimand therefore equals the ATE if and only if treatment effects ($\Delta_i$) are
uncorrelated with responsiveness ($p_i(1)-p_i(0))$: the individuals whom assignment moves most
strongly into treatment must be, on average, neither higher- nor lower-gain
than the population. When they are higher-gain, the estimand exceeds the ATE;
when lower-gain, it falls short. Uniform responsiveness, $p_i(1)-p_i(0)$
constant across individuals, is the leading special case in which the
covariance vanishes automatically.

\subsection{Binary instrumental variables}

    The binary IV setting is the special case in which $Z_i\in\{z,z'\}$ is interpreted as an instrument. The Wald estimand is
\begin{equation} \label{eq:wald}
        \beta^W
        =\frac{\E[Y_i\mid Z_i=z']-\E[Y_i\mid Z_i=z]}{\E[D_i\mid Z_i=z']-\E[D_i\mid Z_i=z]}.
\end{equation}
Under instrument independence, exclusion, no-within selection and relevance (i.e., $\E[D_i\mid Z_i=z']-\E[D_i\mid Z_i=z] \ne 0$),
\[
\beta^W
=
\frac{\E[\Delta_i\{p_i(z')-p_i(z)\}]}
     {\E[p_i(z')-p_i(z)]}
=
\E\!\left[
\Delta_i
\frac{p_i(z')-p_i(z)}
     {\E[p_i(z')-p_i(z)]}
\right].
\]
Thus IV identifies an instrument-responsiveness-weighted treatment effect as previously noted, among others, in \cite{dinardolee2011}\footnote{This representation also appears in \cite{small2017}, who motivate stochasticity from unrepresented versions of the IV and proxy instruments (see further references in the article), thus focussing on measurement and design issues rather than behavior.  They define ``stochastic monotonicity'' in terms of strata based on all measured and unmeasured confounders of the relationship between treatment ant outcome: ``the weight for a group of subjects is proportional to the size of the group times how much higher the probability of taking the treatment for subjects in the group when given the encouraging level of the IV is compared to the non-encouraging level'' (p.563)}.

The deterministic compliance model is nested as a special case. If each individual has deterministic potential treatment choices $D_i(z'),D_i(z)\in\{0,1\}$, then
\[
        p_i(z')=D_i(z'),
        \qquad
        p_i(z)=D_i(z),
\]
and
\[
        r_i(z',z)=D_i(z')-D_i(z).
\]
If monotonicity holds, $D_i(z')\geq D_i(z)$ for all $i$, then
\[
        r_i(z',z)=\1\{D_i(z')>D_i(z)\}.
\]
Equation \eqref{eq:wald} becomes
\[
        \beta^W=\E[\Delta_i\mid D_i(z')>D_i(z)],
\]
the standard Local Average Treatment Effect (LATE) for compliers.

Under stochastic potential treatments or choices, there need not be a discrete complier group. If $p_i(z')\geq p_i(z)$ for all $i$, the Wald estimand is still a convex average (see Section \ref{sec:monotonicity}), but the weights are proportional to continuous responsiveness. Even if one interprets $p_i(z')\geq p_i(z)$ as a relaxation of the conventional monotonicity condition towards a stochastic monotonicity one and classifies those for whom $p_i(z')> p_i(z)$ as ``stochastic compliers,'' the estimand corresponds to a weighted average among those individuals with possibly unequal weights that are proportional to individual responsiveness $r_i(z',z)$.  
If $p_i(z')-p_i(z)$ takes both signs, the estimand is not an average. It is a signed contrast.

Following the same steps, one can also obtain stochastic analogues of the
``kappa weights'' introduced by \cite{abadie2003} (see his Theorem~3.1),
which in the deterministic potential-outcome model provide ``a weighting
scheme that allows us to identify expectations for compliers'' (p.~237).
For this discussion, relabel the two instrument states as
\(Z_i\in\{0,1\}\), define
\[
        \pi(X_i)=P(Z_i=1\mid X_i),
        \qquad
        r_i=p_i(1)-p_i(0),
\]
and suppose \(0<\pi(X_i)<1\) almost surely. Maintain the conditional version
of Assumption~\ref{ass:stateindependence} and
Assumption~\ref{exclusion}.

To obtain weighting identities for arbitrary transformations of the outcome,
we require a distributional strengthening of
Assumption~\ref{ass:separability}. Specifically, for every measurable set
\(A\), state \(z\), and treatment value \(d\) such that
\(p_i(d\mid z)>0\), suppose
\[
\Pr\{Y_i(d,z)\in A\mid G_i,D_i(z)=d\}
=
\Pr\{Y_i(d,z)\in A\mid G_i\}
=
Q_i(A\mid d,z).
\]
This condition is stronger than Assumption~\ref{ass:separability}, which
restricts only conditional means, but it is needed here only because the
weighting identities below are stated for arbitrary measurable functions of
the outcome. The residual-shock independence condition in
footnote~\ref{assumption1} is sufficient. For the special case
\(g(Y,X)=Y\) used below, the original mean-separability assumption is
sufficient.

Define
\[
\kappa_{1i}
=
D_i\,\frac{Z_i-\pi(X_i)}
{\pi(X_i)\{1-\pi(X_i)\}},
\qquad
\kappa_{0i}
=
(1-D_i)\,
\frac{\pi(X_i)-Z_i}
{\pi(X_i)\{1-\pi(X_i)\}},
\]
and
\[
\kappa_i
=
\kappa_{0i}\{1-\pi(X_i)\}
+
\kappa_{1i}\pi(X_i)
=
1-\frac{D_i(1-Z_i)}{1-\pi(X_i)}
-\frac{(1-D_i)Z_i}{\pi(X_i)}.
\]

For any measurable \(h\) satisfying
\(E|h(Y_i,X_i)|<\infty\), define the type-level expectations
\[
\bar h_{i,d}(X_i)
=
\int h(y,X_i)\,Q_i(dy\mid d),
\qquad d\in\{0,1\}.
\]
Then
\begin{align*}
E[\kappa_{1i}h(Y_i,X_i)]
&=
E\!\left[r_i\,\bar h_{i,1}(X_i)\right],\\
E[\kappa_{0i}h(Y_i,X_i)]
&=
E\!\left[r_i\,\bar h_{i,0}(X_i)\right].
\end{align*}
More generally, for an integrable measurable function \(g(y,d,x)\), define
\[
\bar g_{i,d}(X_i)
=
\int g(y,d,X_i)\,Q_i(dy\mid d).
\]
Using
\(
\kappa_i
=
\{1-\pi(X_i)\}\kappa_{0i}
+
\pi(X_i)\kappa_{1i},
\)
we obtain
\[
E[\kappa_i g(Y_i,D_i,X_i)]
=
E\!\left[
r_i
\left\{
\pi(X_i)\bar g_{i,1}(X_i)
+
\{1-\pi(X_i)\}\bar g_{i,0}(X_i)
\right\}
\right].
\]
In particular, taking \(g\equiv1\) gives
\[
E[\kappa_i]=E[r_i].
\]

These identities follow directly from conditional state independence. For
example,
\[
E[\kappa_{1i}h(Y_i,X_i)\mid G_i]
=
\{p_i(1)-p_i(0)\}\bar h_{i,1}(X_i),
\]
because, conditional on \(G_i\),
\(P(Z_i=1\mid G_i)=\pi(X_i)\), while the strengthened
within-type separability condition allows the observed treated outcome
distribution to be replaced by \(Q_i(\cdot\mid1)\). The identity for
\(\kappa_{0i}\) is analogous.

The deterministic potential-outcome model is nested as a special case. If
\(p_i(z)=D_i(z)\in\{0,1\}\), deterministic monotonicity holds,
\(D_i(1)\geq D_i(0)\), and the outcome kernels are degenerate at the
potential outcomes \(Y_i(d)\), then
\[
r_i
=
D_i(1)-D_i(0)
=
\mathbf 1\{D_i(1)>D_i(0)\},
\]
and
\[
\bar h_{i,d}(X_i)=h(Y_i(d),X_i).
\]
The normalized identities above then reduce to the familiar Abadie
weighting formulas for compliers.

In the stochastic model the same observable weights instead recover
responsiveness-weighted type-level expectations. In particular, if
\(E[r_i]\neq0\), setting \(h(Y,X)=Y\) gives
\[
\frac{E[\kappa_{1i}Y_i]}{E[\kappa_i]}
-
\frac{E[\kappa_{0i}Y_i]}{E[\kappa_i]}
=
\frac{E[r_i\{\mu_i(1)-\mu_i(0)\}]}{E[r_i]}
=
E\!\left[
\Delta_i
\frac{p_i(1)-p_i(0)}
     {E[p_i(1)-p_i(0)]}
\right].
\]
Thus the kappa weighting construction continues to use functions of the
observable data, but under stochastic treatment choice its latent
interpretation changes: rather than selecting a discrete complier
subpopulation, it weights type-level outcome distributions in proportion to
the treatment-probability movement induced by the instrument.

Recently, \cite{navjeenanetal2026} show, in a deterministic potential outcomes setting, that the existence of a version of such kappa weights is not only sufficient but also necessary for the identification of certain causal parameters in more general models.  We conjecture that the results there can be generalized to the stochastic potential treatments and outcomes framework as well.

\subsection{Multivalued Instruments, 2SLS and Stochastic Potential Outcomes}
\label{subsec:tsls}

The binary Wald estimand is easy to interpret because the instrument defines a single
comparison: state \(z'\) versus state \(z\). Most empirical IV applications are not quite
so simple. The instrument may be multivalued or continuous, the researcher may use
several excluded instruments, and the specification may include covariates. We consider here the well known two stage least squares estimand and interpret it in light of the stochastic potential outcome formulation we are considering. 

Consider the linear IV specification with one endogenous treatment,
\begin{equation}\label{eq:2sls}
Y_i = \alpha + \beta D_i + X_i^\top\gamma + u_i,    
\end{equation}
estimated using excluded instruments constructed from \(Z_i\). 
Equation~(\ref{eq:2sls}) is the estimating equation that defines the 2SLS
procedure, not a generative model for the outcome: the residual $u_i$ is defined by the
associated projection,\footnote{In other words, $u_i$ is such that $E(u_i)=0$, $E(u_i X_i)=0$ and $E(u_i L(D_i|B_i,X_i))=0$, where $L(D_i|B_i,X_i)$ is the linear projection of $D_i$ on $X_i$ and $B_i$ defined in what follows.} no restriction such as $\E[u_i\mid X_i,Z_i]=0$ is
imposed, and $\beta$ is not assumed to be a common structural coefficient.
We are interested in interpreting the  population 2SLS functional (defined more precisely below)
$\beta^{2SLS}=\E[D_i^{\mathrm{fs}}Y_i]/\E[D_i^{\mathrm{fs}}D_i]$ under
Assumptions~\ref{ass:separability}--\ref{exclusion}, where $D_i^{\mathrm{fs}}$ is defined below.

Let
\[
B_i=b(Z_i,X_i)
\]
denote the vector of excluded instrument terms used in the first stage. For example,
\(B_i\) may contain \(Z_i\) itself, dummies for values of a multivalued instrument,
polynomials in a continuous instrument, or interactions between the instrument and
covariates.

For any random variable \(V_i\), let

\[
\ddot V_i = V_i - L[V_i\mid X_i]
\]
denote the residual obtained after removing its linear projection on $X_i$, which we denote by $L[V_i\mid X_i]$ (here we require that all projections include a constant term).
In particular,

\[
\ddot B_i = B_i - L[B_i\mid X_i],
\qquad
\ddot D_i = D_i - L[D_i\mid X_i],
\quad \textrm{and} \quad
\ddot Y_i = Y_i - L[Y_i\mid X_i].
\]

The (residualized) population first stage is
\[
\ddot D_i = \ddot B_i^\top\pi_D + v_i,
\qquad
\pi_D
=
E[\ddot B_i\ddot B_i^\top]^{-1}E[\ddot B_i\ddot D_i],
\]
assuming the usual rank condition. Define
\[
D_i^{\mathrm{fs}}=\ddot B_i^\top \pi_D,
\]
which is the part of treatment receipt predicted by the excluded instruments after
partialling out the included controls. This is the population analogue of the usual
Frisch--Waugh--Lovell residualization.  With one endogenous treatment, even an
overidentified 2SLS estimator can be viewed as using this single fitted first-stage
component as the relevant IV variation. The population 2SLS coefficient is therefore
\[
\beta^{2SLS}
=
\frac{E[D_i^{\mathrm{fs}}\ddot Y_i]}
     {E[D_i^{\mathrm{fs}}\ddot D_i]}
=
\frac{E[D_i^{\mathrm{fs}}Y_i]}
     {E[D_i^{\mathrm{fs}}D_i]},
\]
where the second equality uses \(E[D_i^{\mathrm{fs}} L(D_i | X_i)]=0\) and \(E[D_i^{\mathrm{fs}} L(Y_i | X_i)]=0\) since, by definition, $E(\ddot V_i X_i)=0$ for any variable $V_i$.

Now impose the stochastic potential outcomes Assumptions \ref{ass:separability}-\ref{exclusion} from the previous section: no within-type
selection, (conditional) state independence, and exclusion. Then\footnote{As a reminder, and as stated in Section 2 above, $X_i$ is measurable with respect to $G_i$ (covariates are part of the type), so conditioning on $X_i$ and $G_i$ is the same as conditioning on $G_i$ alone.}
\[
E[Y_i\mid G_i,X_i,Z_i]
=
\mu_i(0)+p_i(Z_i)\Delta_i,
\quad \textrm{and} \quad
E[D_i\mid G_i,X_i,Z_i]
=
p_i(Z_i),
\]
where again the individual mean treatment effect $\Delta_i$ is
\[
\Delta_i=\mu_i(1)-\mu_i(0).
\]
Assume then that 
\[
L[B_i\mid X_i] = E[B_i \mid X_i].
\]
In other words: the conditional expectation function of $B_i$ given $X_i$ is linear.  This simplifies the expressions that follow and is reasonable if $X_i$ is discrete or a sufficiently rich function of observed covariates, for example. This implies that $E[\ddot B_i \mid X_i]=0$ and thus
\[
E[D_i^{\mathrm{fs}}\mid X_i]=0.
\]
We then have that
\begin{align*}
E[D_i^{\mathrm{fs}} \mu_i(0)] & = E\{E[D_i^{\mathrm{fs}} \mu_i(0) \mid X_i,G_i]\} = E\{E[D_i^{\mathrm{fs}} \mid X_i,G_i] \mu_i(0) \} = E\{E[D_i^{\mathrm{fs}} \mid X_i] \mu_i(0) \} = 0,
\end{align*}
and the baseline potential outcome term $E[D_i^{\mathrm{fs}} \mu_i(0)]$ drops out of the numerator in $\beta^{2SLS}$ above.  The first equality above uses the Law of Iterated Expectations. The second equality obtains from $Q_i(\cdot\mid 0,z)=Q_i(\cdot\mid 0), \forall z \in \cZ$ (Assumption \ref{exclusion}) and thus $\mu_i(0)=\int y\,Q_i(\dd y\mid 0)$ is a direct function of $G_i$.  Finally, the third equality relies on \(Z_i\Perp G_i\mid X_i\) by Assumption \ref{ass:stateindependence} and the last equality holds given that linear conditional expectation assumption above.  Hence:
\[
E[D_i^{\mathrm{fs}}Y_i]
=
E\!\left[D_i^{\mathrm{fs}}\{\mu_i(0)+p_i(Z_i)\Delta_i\}\right]
=
E\!\left[\Delta_i
E\{D_i^{\mathrm{fs}}p_i(Z_i)\mid G_i,X_i\}\right].
\]
Similarly,
\[
E[D_i^{\mathrm{fs}}D_i]
=
E\!\left[E\{D_i^{\mathrm{fs}}p_i(Z_i)\mid G_i,X_i\}\right].
\]

Define
\[
\kappa_i^{2SLS}
\equiv
E\{D_i^{\mathrm{fs}}p_i(Z_i)\mid G_i,X_i\}.
\]
This is the individual's contribution to the population first stage generated by the
particular 2SLS specification. To interpret it, hold fixed the individual's type and
covariates ($G_i$ and $X_i$), let the instrument vary according to its conditional distribution, and ask
how the person's treatment probability covaries with the first-stage fitted value used
by 2SLS.

The 2SLS estimand is therefore
\[
\beta^{2SLS}
=
\frac{E[\Delta_i\kappa_i^{2SLS}]}
     {E[\kappa_i^{2SLS}]}.
\]
Equivalently, writing
$
W_i^{2SLS}
=
\frac{\kappa_i^{2SLS}}{E[\kappa_i^{2SLS}]}
$,
we have
\[
\beta^{2SLS}=E[\Delta_i W_i^{2SLS}].
\]

This expression nests the binary Wald case. Suppose there are no covariates and
\(Z_i\in\{0,1\}\). If the first stage uses \(B_i=Z_i\), then
\(D_i^{\mathrm{fs}}\) is proportional to \(Z_i-E[Z_i]\). Thus
\[
\kappa_i^{2SLS}
\propto
\Pr(Z_i=1)\Pr(Z_i=0)\{p_i(1)-p_i(0)\}.
\]
The factor \(\Pr(Z_i=1)\Pr(Z_i=0)\) is common across individuals, so the normalized
2SLS weights are proportional to
\[
p_i(1)-p_i(0),
\]
exactly as in the binary IV estimand in Section 4.2.

With covariates, the same binary-instrument specification gives
\[
\kappa_i^{2SLS}
\propto
\Pr(Z_i=1\mid X_i)\Pr(Z_i=0\mid X_i)\{p_i(1)-p_i(0)\}.
\]
Thus covariate-adjusted 2SLS weights individuals not only by how much the instrument
changes their treatment probability, but also by how much residual instrument variation
is available in their covariate cell.

With a multivalued instrument, there is generally no single pairwise comparison. If
\(Z_i\) is discrete (but not necessarily binary), then
\[
\kappa_i^{2SLS}
=
\sum_{z\in\mathcal Z}
D^{\mathrm{fs}}(z,X_i)\Pr(Z_i=z\mid X_i)p_i(z),
\]
where \(D^{\mathrm{fs}}(z,x)\) is the first-stage fitted component when the instrument
state is \(z\) and covariates are \(x\). Since
\[
\sum_{z\in\mathcal Z}
D^{\mathrm{fs}}(z,X_i)\Pr(Z_i=z\mid X_i)=0,
\]
for any baseline state \(z_0\),
\[
\kappa_i^{2SLS}
=
\sum_{z\in\mathcal Z}
D^{\mathrm{fs}}(z,X_i)\Pr(Z_i=z\mid X_i)
\{p_i(z)-p_i(z_0)\}.
\]
Therefore, 2SLS with a multivalued instrument is a weighted combination of the
treatment-probability movements generated by the different instrument states. Using the
instrument in levels, using a set of dummies, or interacting the instrument with
covariates generally changes \(D_i^{\mathrm{fs}}\), and therefore changes the weights in
the treatment-effect average.

The main implication is that, beyond the binary Wald case, instrument validity does not
by itself determine the treatment-effect estimand\footnote{In the deterministic-compliance benchmark, \cite{sloczynski2020should} makes a closely related point: once covariates are needed for identification, the standard non-interacted linear IV specification can place negative weights on conditional LATEs — and lose its causal interpretation — while the saturated, fully interacted specification of \cite{angristimbens1995} restores convex weights (see also \cite{kolesar2013} and \cite{blandholetal2026}). Our $\kappa_i^{2SLS}$ is the stochastic potential outcomes counterpart of the object he studies, and moving $B_i$ from levels to a saturated set of interactions is precisely the lever that alters its sign.}. The estimand also depends on the
first-stage specification: which functions of the instrument are used, how covariates are
partialled out, and how multiple instruments are combined by 2SLS. This is the
stochastic-choice analogue of the familiar LATE lesson that IV estimates are local to the
margin moved by the instrument. Here the margin is not necessarily a discrete complier
group even when there is (stochastic) monotonicity. It is the distribution of individual treatment-probability movements induced by
the fitted first stage.

Consequently, two valid 2SLS specifications can identify different parameters. If
specifications \(a\) and \(b\) induce weights \(W_i^a\) and \(W_i^b\), then
\[
\beta^a-\beta^b
=
E[\Delta_i\{W_i^a-W_i^b\}]
=
\operatorname{Cov}\!\left(\Delta_i,W_i^a-W_i^b\right),
\]
because \(E[W_i^a]=E[W_i^b]=1\). As in the deterministic case, differences across valid IV estimates may therefore
reflect different first-stage margins rather than failure of exclusion or independence.

\section{Monotonicity and the Average Interpretation}\label{sec:monotonicity}

The Wald ratio is a ratio of two causal effects of the state: the reduced form and the
first stage. Under exclusion, this ratio can be written as
\[
\theta(z',z)
=
\frac{E[\Delta_i r_i(z',z)]}{E[r_i(z',z)]},
\qquad
r_i(z',z)=p_i(z')-p_i(z),
\]
where \(r_i(z',z)\) is the change in individual \(i\)'s probability of treatment induced by
the state contrast.

The econometric question is not only whether this ratio is well defined. It is whether
the ratio can be interpreted as an average of treatment effects. In the deterministic
LATE model, monotonicity guarantees that. It rules out defiers, makes the first
stage equal to the probability of switchers, and turns the Wald estimand into the average
treatment effect for compliers. In the stochastic model, the switching indicator ($\mathbf{1}\{D_i(z')>D_i(z)\}$) is
replaced by the probability change \(r_i(z',z)\). Hence the analogue of monotonicity is
one-sided responsiveness:
\[
r_i(z',z)\geq 0 \quad \text{for all } i,
\]
or, after relabeling the two states, \(r_i(z',z)\leq 0\) for all \(i\).

If responsiveness is one-sided and the first stage is non-zero, then
\[
\theta(z',z)
=
E\!\left[
\Delta_i
\frac{r_i(z',z)}{E[r_i(z',z)]}
\right],
\]
with nonnegative weights ($r_i(z',z)/E[r_i(z',z)]$) that integrate to one. The IV ratio is therefore a convex
responsiveness-weighted average of individual mean treatment effects. Individuals whose
treatment probabilities move more receive more weight. Individuals whose treatment
probabilities do not move receive zero weight.

This nests the standard cases. If
\[
r_i(z',z)=c\neq 0
\]
for all \(i\), then the ratio equals the ATE. When treatment choice is deterministic and
monotonic, then
\[
r_i(z',z)=D_i(z')-D_i(z)
=
\mathbf{1}\{D_i(z')>D_i(z)\},
\]
and the ratio becomes the usual LATE:
\[
\theta(z',z)
=
E[\Delta_i\mid D_i(z')>D_i(z)].
\]
For stochastic, but one-sided, treatment choice, the same logic survives, except that the complier indicator is replaced by the amount of probabilistic responsiveness.

The difficulty arises when responsiveness changes sign. Suppose, for simplicity, that
\(E[r_i(z',z)]>0\). Define
\[
A^+
=
E\!\left[r_i(z',z)\mathbf{1}\{r_i(z',z)>0\}\right],
\qquad
A^-
=
E\!\left[-r_i(z',z)\mathbf{1}\{r_i(z',z)<0\}\right],
\]
\[
\theta^+
=
\frac{
E\!\left[\Delta_i r_i(z',z)\mathbf{1}\{r_i(z',z)>0\}\right]
}{
E\!\left[r_i(z',z)\mathbf{1}\{r_i(z',z)>0\}\right]
},
\quad \textrm{and} \quad
\theta^-
=
\frac{
E\!\left[\Delta_i\{-r_i(z',z)\}\mathbf{1}\{r_i(z',z)<0\}\right]
}{
E\!\left[\{-r_i(z',z)\}\mathbf{1}\{r_i(z',z)<0\}\right]
}.
\]
Then
\[
\theta(z',z)
=
\frac{A^+\theta^+-A^-\theta^-}{A^+-A^-}.
\]
Thus the Wald ratio compares treatment effects among types whose treatment probabilities
rise with treatment effects among types whose treatment probabilities fall. The denominator
is the net first stage, \(A^+-A^-(= E[r_i(z',z)])\), not the total amount of treatment-probability (absolute value) movement,
\(A^++A^- (= E[|r_i(z',z)|])\). Importantly, when positive and negative movements partly offset, the ratio can lie outside the range of individual treatment effects.

\begin{proposition}[Average interpretation and one-sided responsiveness]
\label{prop:onesided}
Suppose individual mean treatment effects lie in a nondegenerate bounded
interval $[\underline{\Delta}, \overline{\Delta}]$, with
$-\infty < \underline{\Delta} < \overline{\Delta} < \infty$, and suppose
$E[r_i(z',z)] \neq 0$. The ratio
\[
\theta(z',z) \;=\; \frac{E[\Delta_i\, r_i(z',z)]}{E[r_i(z',z)]}
\]
is guaranteed to lie in $[\underline{\Delta}, \overline{\Delta}]$ for every
pattern of treatment-effect heterogeneity with
$\Delta_i \in [\underline{\Delta}, \overline{\Delta}]$ if and only if
$r_i(z',z)$ has one sign almost surely.
\end{proposition}

\begin{proof}
($\Rightarrow$) If $r_i(z',z) \geq 0$ almost surely and $E[r_i(z',z)] > 0$,
the normalized weights $r_i(z',z)/E[r_i(z',z)]$ are nonnegative and integrate
to one, so the ratio is a convex average of values in
$[\underline{\Delta}, \overline{\Delta}]$ and therefore lies in that interval.
The case $r_i(z',z) \leq 0$ almost surely is identical after multiplying
numerator and denominator by $-1$.

($\Leftarrow$) Suppose instead that $r_i(z',z)$ takes both positive and
negative values with positive probability. Recall
\[
A^+ = E\big[r_i(z',z)\,\mathbf{1}\{r_i(z',z) > 0\}\big],
\qquad
A^- = E\big[-r_i(z',z)\,\mathbf{1}\{r_i(z',z) < 0\}\big],
\]
so that $A^+ > 0$, $A^- > 0$, and $E[r_i(z',z)] = A^+ - A^-$. Consider a population where
\[
\Delta_i =
\begin{cases}
\underline{\Delta} & \text{if } r_i(z',z) > 0,\\[2pt]
\overline{\Delta} & \text{if } r_i(z',z) \leq 0,
\end{cases}
\]
which satisfies $\Delta_i \in [\underline{\Delta}, \overline{\Delta}]$. Then
$E[\Delta_i\, r_i(z',z)] = \underline{\Delta}A^+ - \overline{\Delta}A^-$, and
direct computation gives the two identities
\[
\theta(z',z) - \underline{\Delta}
= -\,\frac{(\overline{\Delta}-\underline{\Delta})\,A^-}{E[r_i(z',z)]},
\qquad
\theta(z',z) - \overline{\Delta}
= -\,\frac{(\overline{\Delta}-\underline{\Delta})\,A^+}{E[r_i(z',z)]}.
\]
Since $\overline{\Delta} - \underline{\Delta} > 0$ and $A^+, A^- > 0$, the
first identity implies $\theta(z',z) < \underline{\Delta}$ when
$E[r_i(z',z)] > 0$, and the second implies
$\theta(z',z) > \overline{\Delta}$ when $E[r_i(z',z)] < 0$. In either case
the ratio lies strictly outside $[\underline{\Delta}, \overline{\Delta}]$.
Thus, when responsiveness changes sign, the ratio cannot be guaranteed to
have a convex-average interpretation.
\end{proof}

The proposition clarifies the role of monotonicity. In the deterministic IV model,
monotonicity rules out defiers. In the stochastic model, one-sided responsiveness rules
out types whose treatment probabilities move in the opposite direction. The result above is also noted in \cite{dinardolee2011}.  Without this
condition, the reduced form remains a causal effect of the assigned state, but the Wald
ratio is a signed treatment-effect contrast rather than an average treatment effect.

The same point applies to more general IV estimands. For two-stage least squares with
covariates or multiple excluded instruments, the relevant object is the individual first-stage
contribution induced by the fitted first stage. The 2SLS estimand has an average-treatment-effect
interpretation only when those induced first-stage contributions are one-sided. Otherwise,
valid instruments may still identify signed contrasts rather than convex averages.

\section{An Information Foundation for Stochastic Choice}
\label{sec:information}

The preceding sections develop a stochastic potential choice and outcome
model and interpret common estimands without specifying why treatment choice
remains stochastic even within individual types. As noticed elsewhere, experimental evidence suggests that individual choices may vary even when confronted with the same decision scenario (see, e.g., \cite{strzalecki2025}, p.9). This may reflect transient frictions and fluctuations in perceived utility, learning and experimentation, random consideration or attention, genuine implementation or decision errors or even deliberate randomization (see again \cite{strzalecki2025}, pp.9-10). Given this, ``[i]t will be
realistic to assume that a man's choice from a given set of alternatives is
not unique but obeys some probability distribution''
\citep[p.~312]{marschak1960}.

Formally, we model this by postulating an information set (or sigma-field) $\mathcal{F}_i$ for individual $i$ where \(M_i(z)\) denotes the systematic component of individual
\(i\)'s perceived private net value of treatment under state \(z\), and
\(\sigma_i(z)>0\) measures the dispersion of the part of the perception error
that remains unresolved at the moment of choice under that state. Whereas both \(M_i(z)\) and \(\sigma_i(z)>0\) pertain to the information set $\mathcal{F}_i$, residual variation ascribed to a perception error $\varepsilon_i$ does not.  The state $z$ here
may represent policy, default rule, price, subsidy, information
treatment, or another feature of the choice environment. The 
treatment margin determining choices is then given by the random index
\begin{equation}
\widetilde M_i(z)
=
M_i(z)+\sigma_i(z)\,\varepsilon_i,
\label{eq:perceived_margin}
\end{equation}
where \(\varepsilon_i\) is a perception error with a symmetric, continuous
distribution function \(F\), density \(f>0\), mean zero and unit variance,
drawn once conditional on \(G_i\) and, in line with the discussion following
Assumption~\ref{ass:stateindependence}, independent of \(Z_i\) given
\(G_i\).\footnote{The selection is related to that in \cite{klein2011}, who studies (small random) deviations (via $\varepsilon_i$) from (deterministic) monotonicity. The interpretation and purpose of the idiosyncratic term $\varepsilon_i$ is thus distinct.} Under this normalization
\(\sigma_i^2(z)=\Var\{\widetilde M_i(z)\mid G_i\}\). The functions
\(M_i(\cdot)\) and \(\sigma_i(\cdot)\), together with \(F\), are stable
objects and constitute a structural parameterization of the treatment choice
kernel in \(G_i\); the realization of \(\varepsilon_i\) is not part of the
type and is the residual randomness that makes choice stochastic (within
type). Here, a more informative state,
formalized as a smaller \(\sigma_i(z)\), means that a larger share of the index (\ref{eq:perceived_margin}) is ascribed to the deterministic component \(M_i(\cdot)\).\footnote{A multi-component analogue of the sequential
information structure in \citet[Section VI]{heckmannavarro} makes the
resolution structure explicit. Let \(\mathbf u_i=(u_{i1},\ldots,u_{iK})\) be
 mean-zero variables affecting the choice drawn given \(G_i\),
and let the state $z$ determine which components the individual learns before choosing.  Denote by \(J_i(z)\subseteq\{1,\ldots,K\}\) the set of indices for those variables in individual $i$'s information set.  The unresolved part of the choice index is
\(\varepsilon_i(z)=\sum_{k\notin J_i(z)}u_{ik}\). The scalar representation \eqref{eq:perceived_margin} is the
marginal, location--scale summary of such constructions and is all that is
used below: every estimand in
Sections~\ref{sec:observedcontrasts}--\ref{sec:monotonicity} depends on the
type only through the marginal kernels \(p_i(\cdot\mid z)\), so the joint
law of \(\{\widetilde M_i(z):z\in\cZ\}\) across states matters for
interpretation rather than identification (see also
Appendix~\ref{app:equivalence}).} Consistency with
Assumption~\ref{ass:separability} (no within-type selection on residual outcome shocks) requires that the 
 choice assessment error be conditionally mean-independent of 
residual outcome shocks:
\[
\E\{Y_i(d,z)\mid G_i,\varepsilon_i\}=\mu_i(d,z).
\]

If treatment choice or allocation follows a threshold-crossing rule with zero as the threshold (i.e., $D_i(z)=1$ if and only if $\widetilde M_i(z)\geq 0$), the treatment choice kernel is given by:
\begin{equation}
p_i(z)
=
\Pr\{\widetilde M_i(z)\geq 0\mid G_i\}
=
F\big(t_i(z)\big)
\qquad \textrm{with} \qquad
t_i(z)\equiv\frac{M_i(z)}{\sigma_i(z)},
\label{eq:choice_curve}
\end{equation}
where we assume that \(F\) is continuous and corresponds to a symmetric distribution (around zero).
We refer to the map \(t\mapsto F(t)\) as the individual \emph{choice curve}
and to \(t_i(z)\) as the \emph{standardized perceived margin}. Everything
the individual contributes to the estimands of the previous sections is her
position on this curve under each state. Before the perception error is
realized, \(p_i(z)\) is also the individual's own ex ante choice
probability and corresponds to the object elicited in stated-probability data
\citep{blasslachmanski2010,arcidiaconoetal2020,briggsetal2024}.
Deterministic compliance is nested as a limit: as \(\sigma_i(z)\to0\) and if \(\Pr(M_i(z)=0)=0\), the
choice curve steepens to a step function, \(p_i(z)\to\1\{M_i(z)\ge 0\}\), and
the deterministic taxonomy of Section~\ref{sec:familiar}
re-emerges.\footnote{More generally, one may let the unresolved error have
an arbitrary state-dependent conditional law \(F_{i,z}\) given \(G_i\), so
that \(p_i(z)=1-F_{i,z}(-M_i(z))\) and, for a continuously indexed scalar
state,
\(\dot p_i(z)=f_{i,z}(-M_i(z))\,\dot M_i(z)
-\partial_z F_{i,z}(u)\big|_{u=-M_i(z)}\).
The location--scale family is the only case used in this section; it
already spans the two channels through which a state can act on choice,
by shifting the location of the perceived margin or by changing its
scale.}

Suppose now that \(z\) is a continuously indexed scalar state and that
\(M_i(z)\) and \(\sigma_i(z)\) are differentiable in \(z\). Differentiating
\eqref{eq:choice_curve} gives a measure of local responsiveness:
\begin{equation}
\dot p_i(z)
=
\frac{f\big(t_i(z)\big)}{\sigma_i(z)}
\Big\{\dot M_i(z)-t_i(z)\,\dot\sigma_i(z)\Big\}.
\label{eq:two_channels}
\end{equation}
Local responsiveness is a nonnegative factor times the difference between two terms. The
factor \(f(t_i(z))/\sigma_i(z)\) is the density of the treatment margin in (\ref{eq:perceived_margin}) at
zero.  This term is large when the individual is close to
indifference and the distribution $F$ is not only symmetric, but also single-peaked at zero. The first term, \(\dot M_i(z)\), is a \emph{location}
channel: the state shifts the systematic, deterministic component of the treatment margin,
moving the individual along her choice curve. The second one,
\(-t_i(z)\,\dot\sigma_i(z)\), is a \emph{precision} channel: the state
changes how much of the perception error contributes to the treatment allocation, steepening or
flattening the curve. Because \(p_i(z)\) depends on
\((M_i(z),\sigma_i(z))\) only through the ratio \(t_i(z)\), the two channels
are two ways of moving the standardized margin: a location shift moves it
additively, while a precision change rescales it multiplicatively. This
distinction governs monotonicity below. Figure~\ref{fig:choice_curve} illustrates the two channels.

\begin{figure}[t]
\centering
\pgfplotsset{choicecurve/.style={width=0.46\textwidth, height=0.36\textwidth,
    xmin=-4.5, xmax=4.5, ymin=-0.02, ymax=1.02,
    xtick={0}, ytick={0,0.5,1}, axis lines=left,
    xlabel={$M_i$}, ylabel={$p$}, samples=120, domain=-4.5:4.5}}
\begin{tikzpicture}
\begin{axis}[choicecurve, title={Location channel}]
  \addplot[thick]         {1/(1+exp(-x))};          
  \addplot[thick, dashed] {1/(1+exp(-(x+1.5)))};    
\end{axis}
\end{tikzpicture}\hfill
\begin{tikzpicture}
\begin{axis}[choicecurve, title={Precision channel}]
  \addplot[thick]         {1/(1+exp(-x))};          
  \addplot[thick, dashed] {1/(1+exp(-x/0.45))};     
\end{axis}
\end{tikzpicture}
\caption{Two channels of a state acting on the choice probability
$p_i(z)=F\{M_i(z)/\sigma_i(z)\}$, plotted against the systematic margin $M_i$
(baseline scale held fixed). Solid: before the intervention; dashed: after.
Left: a mean-shifting state, $M_i \mapsto M_i+\ell_i\,\Delta z$ with $\sigma_i$
fixed, translates the curve horizontally, so responsiveness is one-sided.
Right: an information refinement, $\sigma_i' < \sigma_i$ with $M_i$ fixed,
steepens the curve, pivoting it around $p=1/2$ at $M_i=0$: take-up rises for
$M_i>0$ and falls for $M_i<0$, so responsiveness is two-sided.}
 \label{fig:choice_curve}
\end{figure}
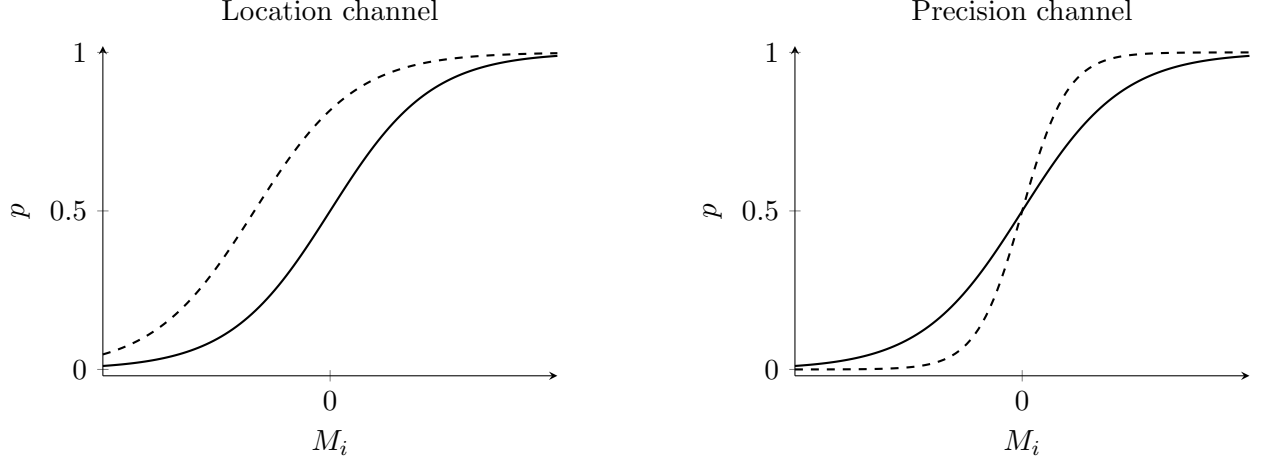
 
For a small change in \(z\), given the Wald estimand in \eqref{eq:wald},
the induced IV-type local parameter is (assuming the relevant limits and
expectations may be interchanged):
\[
\theta(z)
=
\frac{\E[\Delta_i\,\dot p_i(z)]}{\E[\dot p_i(z)]},
\]
whenever the denominator is nonzero. Thus the local weights are proportional
to \(\dot p_i(z)\), the individual-level change in treatment probability
induced by the state. The two subsections that follow examine each channel
in isolation.

\subsection{Mean-Shifting States: Moving Along the Choice Curve}
\label{sec:location}

First consider a state that shifts the systematic perceived margin while
leaving its precision unchanged, so \(\dot\sigma_i(z)=0\). Suppose that
\[
M_i(z)=M_i+\ell_i z,
\qquad
\ell_i\geq 0.
\]
Then \eqref{eq:two_channels} reduces to
\[
\dot p_i(z)
=
\ell_i\,\frac{f\big(t_i(z)\big)}{\sigma_i(z)}
\;\geq\;0,
\]
and, provided \(\E[\ell_i f(t_i(z))/\sigma_i(z)]>0\), the local weights are
\begin{equation}
W_i^{M}(z)
=
\frac{\ell_i\,f\big(t_i(z)\big)/\sigma_i(z)}
     {\E\!\left[\ell_i\,f\big(t_i(z)\big)/\sigma_i(z)\right]}.
\label{eq:location_weight}
\end{equation}
The weight has two components. When the distribution is symmetric and single-peaked at zero, the margin-density term
\(f(t_i(z))/\sigma_i(z)\) is large when the individual is close to
indifference based on their known, perceived margin (i.e., $t_i(z) \approx 0$). The loading \(\ell_i\) measures how
strongly the state shifts that individual's perceived treatment margin.

For a subsidy or cost
instrument $z$, \(\ell_i\) could be interpreted as material pass-through. For a reminder or
salience intervention $z$, \(\ell_i\) could encode attention responsiveness. For a
persuasion or framing instrument $z$, \(\ell_i\) could correspond to susceptibility to the
message. The algebra is the same, but the economic interpretation of the
loading differs.

If \(f\) is single-peaked at zero, then for fixed \(\ell_i\) and
\(\sigma_i\) the unnormalized weight is largest when \(t_i(z)=0\), that is,
at \(p_i(z)=1/2\). Mean-shifting states therefore place the greatest weight
on individuals for whom the systematic component of the treatment margin is near the threshold. Note also
that sign-homogeneous loadings deliver one-sided responsiveness (or monotonicity) by
construction: a location shift translates every standardized margin in the
same direction, \(\dot t_i(z)=\ell_i/\sigma_i(z)\geq 0\), so
\(\dot p_i(z)\geq 0\) for all \(i\) and the local analogue of the
stochastic monotonicity condition of Section~\ref{sec:monotonicity} holds.

\paragraph{Relation to Marginal Treatment Effects (MTE).}
The density term \(f(t_i(z))/\sigma_i(z)\) in \eqref{eq:location_weight} is
the within-type analogue of the object that drives marginal treatment
effects in the latent-index model of \citet{HeckmanVytlacil2005}. There, the choice rule is given by
\(D_i(Z_i)=\1\{\mu_D(Z_i) \geq U_{Di}\}\), 
where $U_D$ can be interpreted a ``resistance to treatment'' term counteracting the state-dependent index \(\mu_D(Z_i)\).  In the Heckman--Vytlacil model the resistance
\(U_D\) is a \emph{fixed} individual scalar type (thus belonging to $\mathcal{F}_i$) and the index \(\mu_D(\cdot)\) is homogeneous
to the population: the instrument moves a single threshold faced by
everyone, the individuals at the margin form the fixed subpopulation
\(\{U_D=\mu_D(z)\}\) and monotonicity follows because the common threshold
moves in one direction for all. This is the configuration under
which the equivalence result established in \citet{vytlacil2002} holds, and it corresponds
here to a location channel with a common index and loading.  In terms of the presentation above, this means that $\sigma_i(z)=0$ for every $z$ (since both $\mu_D(z)$ and $U_D$ are in the agent's information set), $U_{Di}=M_i \in \mathcal{F}_i$ and $\ell_i z = \ell z = \mu_D(z) \in \mathcal{F}_i$ for every $i$.  In the stochastic
model the standardized margin is instead individual-specific: the same
state moves \(t_i(z)\), the standardized loadings
\(\ell_i/\sigma_i(z)\), by different amounts and possibly
in different directions. Moreover, the realization of \(\varepsilon_i\) is
not part of the type as is the realized $M_i \in \mathcal{F}_i$, so the marginal event at \(z\) is a within-type
probabilistic event carrying density \(f(t_i(z))/\sigma_i(z)\), rather than
a fixed subgroup indexed by an individual (realized) scalar resistance. Two consequences follow.
First, effect heterogeneity is indexed by the full type \(G_i\), not by a
single resistance value \(U_{Di}\), so there need not be a real function (i.e., the MTE given by $E(Y_1-Y_0|U_D=u)$) whose (weighted) average corresponds to conventional treatment effect parameters or estimands (e.g., ATE, OLS, IV); see
\citet{HeckmanUrzuaVytlacil2006}. Second, the precision term
\(-t_i(z)\,\dot\sigma_i(z)\) in \eqref{eq:two_channels} interacts the state
with the scale of the decision or perception error and therefore falls outside the
additively separable latent-index class. It is precisely this channel that
delivers the failures of monotonicity documented next.

\subsection{Information Refinements: Steepening the Choice Curve}
\label{sec:precision}

Another possibility is for states to change the precision of the treatment
margin rather than its mean. This can encode simplification,
clarification, disclosure, counseling, or decision aids. Such interventions
refine the information available at the moment of choice reducing the inherent stochasticity: they decrease the role ascribed to the perception error \(\varepsilon_i\), so
the dispersion of its unresolved component falls. Hold the systematic
margin fixed with respect to the state, i.e. \(M_i(z)=M_i\), let $z$ be continuous and suppose that
\[
\dot\sigma_i(z)<0.
\]
Then \eqref{eq:two_channels} reduces to
\begin{equation}
\dot p_i(z)
=
-\,\frac{t_i(z)\,f\big(t_i(z)\big)}{\sigma_i(z)}\,\dot\sigma_i(z)
=
\frac{f\big(t_i(z)\big)\,\lvert\dot\sigma_i(z)\rvert}{\sigma_i(z)}\;t_i(z),
\label{eq:precision_responsiveness}
\end{equation}
so the sign of \(\dot p_i(z)\) is the sign of \(M_i\). Equivalently, a
refinement rescales the standardized margin away from zero,
\(t_i(z')=t_i(z)\cdot\{\sigma_i(z)/\sigma_i(z')\}>t_i(z)\), with the factor $\sigma_i(z)/\sigma_i(z')$
exceeding one when $z<z'$: it pushes each individual further in the direction she
already leans. The choice curve steepens around its midpoint: take-up rises
for individuals whose systematic margin favors treatment
(\(M_i>0\), \(p_i(z)>1/2\)) and falls for individuals whose systematic
margin disfavors it (\(M_i<0\), \(p_i(z)<1/2\)).

\begin{proposition}[Information refinement and two-sided responsiveness]
\label{prop:refinement}
Suppose \eqref{eq:perceived_margin} holds with \(M_i(z)=M_i\), \(f>0\) on
its support, \(\sigma_i(z)>0\), and \(\dot\sigma_i(z)<0\) almost surely. If
\[
\Pr(M_i>0)>0
\qquad\text{and}\qquad
\Pr(M_i<0)>0,
\]
then
\[
\Pr\{\dot p_i(z)>0\}>0
\qquad\text{and}\qquad
\Pr\{\dot p_i(z)<0\}>0.
\]
Hence an information-refining state generically violates local stochastic
monotonicity.
\end{proposition}

\begin{proof}
By \eqref{eq:precision_responsiveness}, \(\dot p_i(z)\) is the product of
the strictly positive factor
\(f(t_i(z))\lvert\dot\sigma_i(z)\rvert/\sigma_i(z)\) and \(t_i(z)\), whose
sign is the sign of \(M_i\). Since \(M_i>0\) and \(M_i<0\) each occur with
positive probability, \(\dot p_i(z)\) takes both signs with positive
probability.
\end{proof}

The failure of monotonicity is exactly what one should expect from better
information. If some individuals mistakenly opt out and others mistakenly
opt in, an information refinement induces entry on one side of the
threshold and exit on the other. Proposition~\ref{prop:refinement}
therefore does not describe a defect in the intervention. Two-sided
responsiveness is the economic content of a policy that reduces mistakes on
both sides of the choice threshold. In the language of
Section~\ref{sec:monotonicity}: stochastic monotonicity is the property
that a state translates every standardized margin in a common direction;
a state that instead rescales the margin generically violates it.

Without monotonicity, the Wald ratio need not be an average treatment
effect. The reduced form, however, remains directly meaningful. Under
exclusion,
\[
m_i(z)=\mu_i(0)+p_i(z)\Delta_i.
\]
Assuming differentiation and expectation can be interchanged,
Proposition~\ref{prop:state_itt} implies
\[
\frac{d}{dz}\E[m_i(z)]
=
\E[\Delta_i\,\dot p_i(z)].
\]
When \(Y_i\) is the outcome relevant to a planner, this derivative is the
effect of refining information on the planner's outcome. Substituting
\eqref{eq:precision_responsiveness} gives
\[
\frac{d}{dz}\E[m_i(z)]
=
\E\!\left[
\underbrace{\frac{f\big(t_i(z)\big)\,\lvert\dot\sigma_i(z)\rvert}
{\sigma_i^2(z)}}_{\geq 0}
\;\Delta_i M_i
\right].
\]
Hence the effect is nonnegative whenever \(\Delta_i M_i\geq 0\) almost
surely; more generally, its sign depends on the weighted average alignment
between the perceived treatment margin \(M_i\) and the true treatment gain
\(\Delta_i\).

As \(\sigma_i(z)\to0\), for every \(M_i\neq0\),
\[
p_i(z)\longrightarrow\1\{M_i>0\}.
\]
Thus, if \(\Pr(M_i=0)=0\), the choice curve converges to a step function
and choice conforms almost surely to the systematic perceived margin. An
information refinement pulls into treatment individuals with \(M_i>0\) and
pushes out individuals with \(M_i<0\): it makes behavior increasingly obey
perceptions. This raises the planner's outcome only to the extent that
perceptions are aligned with treatment gains. If \(M_i>0\) but
\(\Delta_i<0\) for some individuals, more precise information moves them
more confidently toward an action that lowers the planner's outcome. This
can occur because \(M_i\) is a biased perception of the individual's own
payoff, or because the individual's private payoff legitimately differs
from the measured outcome \(Y_i\), in which case the refinement is
privately optimal yet lowers \(Y_i\). Accordingly, the welfare value of a
decision aid depends on the substantive alignment between the margins it
sharpens and the gains relevant to the planner.

\section{Repeated Choices and the Empirical Content of Random Types}
\label{sec:panel}

The preceding sections consider a static cross section scenario, and Appendix~\ref{app:equivalence} shows that the (ex ante) stochastic representation and the (ex post) deterministic representation generate the same observable conditional means in that environment. If restrictions are imposed on how choices are related over time, repeated observations can nonetheless distinguish them. To demonstrate this, consider the stable choice-kernel model
\[
\mathcal M_S:\qquad
D_{it}(z)\mid G_i \ \stackrel{\mathrm{i.i.d.}}{\sim}\
\operatorname{Bernoulli}\{p_i(z)\},
\]
and regard the stable deterministic-response model as its degenerate case
\[
\mathcal M_D:\qquad p_i(z)\in\{0,1\}\quad\text{almost surely for every }z.
\]
Within $\mathcal M_S$, one can use two periods to test whether the choice kernel is degenerate; use short panels to identify moments of the individual treatment probabilities across individuals; and use long panels to identify their distribution. Combining choice paths with outcomes then identifies the joint dependence between treatment effects and the probability movements induced by an instrument or policy. This last object---the latent weighting operator--is the main panel contribution.

The scope of the comparison is important. The results below do not test determinism without restrictions on latent paths $(D_{i1}^{*}(z),\ldots,D_{iT}^{*}(z))$. For any finite $T$, one can draw a latent path $H_i\equiv(D_{i1}^{*}(z),\ldots,D_{iT}^{*}(z))$ from the same distribution as the stochastic choice path and then treat every component of $H_i$ as fixed. This deterministic-path representation has the same observable panel law.
The empirical content comes jointly from stability of the individual choice
kernel, the absence of structural state dependence, and serial independence of
the residual choice draws. Conditional on the stable type \(G_i\), each period
is therefore a new realization of the same static choice problem: the
probability of treatment in period \(t\) depends on the contemporaneous state
only through \(p_i(Z_{it})\), rather than on earlier choices or states. At a fixed state $z$, a finite panel identifies only the first $T$ moments of $p_i(z)$ across individuals; and its distribution is the unique mixing distribution only for an infinitely extendible sequence or in the long-panel limit \citep{diaconisfreedman1980}. We return to the cost of these restrictions in Section~\ref{subsec:panel-discussion}.

\subsection{Setup}
\label{subsec:panel-setup}

For each individual $i$ we now observe a path
\[
O_i \;=\; \{(Y_{it}, D_{it}, Z_{it})\}_{t=1}^{T},
\qquad \mathbf{Z}_i \equiv (Z_{i1},\ldots,Z_{iT}),
\]
with potential treatments $D_{it}(z)$ and potential outcomes $Y_{it}(d,z)$ defined per period. As before, observed variables satisfy $D_{it} = D_{it}(Z_{it})$ and $Y_{it} = Y_{it}(D_{it}, Z_{it})$.  When potential outcomes and treatments are deterministic, the objects
\(D_{it}(\cdot)\) and \(Y_{it}(\cdot)\) are individual-specific functions. In the
stochastic formulation, however, even after fixing the type \(G_i\) and the
state \(z\), treatment receipt may remain random.  Similarly, fixing $G_i$, $d$ and $z$, outcomes may remain random.

\begin{assumption}[Stable types and static resolution]
\label{ass:paneltypes}
There is a time-invariant type $G_i = (X_i, Q_i, p_i)$ such that:
\begin{itemize}
\item[(i)] (\emph{Stationary kernels}) For every $t$, $d$ and $z$,
$\Pr\{D_{it}(z) = 1 \mid G_i\} = p_i(z)$ and $\Pr\{Y_{it}(d,z) \in \cdot \mid G_i\} = Q_i(\cdot \mid d, z)$, with $\mu_i(d,z) = \int y \, Q_i(dy \mid d,z)$ as before (see Definition \ref{def:sttypes});
\item[(ii)] (\emph{Serial independence of residual choice draws}) Conditional on $(G_i, \mathbf{Z}_i)$, the draws $\{D_{it}(Z_{it})\}_{t=1}^{T}$ are independent across $t$ with $\Pr\{D_{it}(Z_{it}) = 1 \mid G_i, \mathbf{Z}_i\} = p_i(Z_{it})$;
\item[(iii)] (\emph{No within-type selection, panel version}) For every $t$, $d$ with $p_i(d \mid Z_{it}) > 0$,
\[
E\{Y_{it}(d, Z_{it}) \mid G_i, \mathbf{Z}_i, D_{i1},\ldots,D_{iT}\} \;=\; \mu_i(d, Z_{it}),
\]
\item[(iv)] (\emph{Strict exogeneity of the state path}) $\mathbf{Z}_i \perp\!\!\!\perp G_i$ (or conditionally on $X_i$, with the obvious modifications below), and, conditional on $G_i$, $\mathbf{Z}_i$ is independent of the residual draws in (ii)--(iii).
\end{itemize}
\end{assumption}

Assumption \ref{ass:paneltypes} is the panel-data version of the stable response type used
throughout the paper. Part (i) imposes stability of the treatment-choice and
outcome kernels over time. Part (ii) rules out structural state dependence and
serial correlation in the residual choice draws: conditional on \(G_i\), each
period represents a new draw from the same choice kernel evaluated at the
current state. Thus, for example, past treatment receipt does not itself change
the current treatment probability once \(G_i\) and the current state are held
fixed. Part (iii) extends Assumption 1 to the panel and rules out residual mean
selection through the realized treatment path. Part (iv) imposes strict
exogeneity of the state path. In particular, states are not assigned in
response to the individual's type or to past residual choice or outcome shocks.

These restrictions are stronger than a no-anticipation condition alone. The
setup is static in the sense that the potential treatment and outcome in a
period are indexed by the contemporaneous state, rather than by the full
history of states and choices. The results below therefore abstract from
learning, structural state dependence, feedback, and other dynamic channels.

\begin{assumption}[Path richness for the contrast $(z',z)$]
\label{ass:pathrich}
Every sequence $\zeta \in \{z, z'\}^{T}$ satisfies $\Pr(\mathbf{Z}_i = \zeta) > 0$.
\end{assumption}

Assumption~\ref{ass:pathrich} is stronger than needed for any single statement below (each result requires positive probability only for the particular sub-paths it uses), but it permits uniform statements. It is naturally satisfied by designs with re-randomized assignment, by rotating defaults or prices, and---approximately---by preference- or examiner-based instruments in which the same individual is repeatedly exposed to differing decision environments.

The deterministic benchmark used below is $\mathcal M_D$: $D_{it}(z)=D_i^{*}(z)$ for every $t$ and every visited state. This time invariance is a substantive restriction, not a consequence of determinism alone. It makes $\mathcal M_D$ the degenerate boundary of the stable-kernel model $\mathcal M_S$.

\subsection{Treatment Choice Data: Kernel Degeneracy and Responsiveness}
\label{subsec:panel-choice}

The first result gives a diagnostic for the degenerate boundary $\mathcal M_D$. Under $\mathcal M_S$, the panel decomposes the cross-sectional variance of treatment into a within-type component and a between-type component,
\begin{equation}
\operatorname{Var}(D_{it} \mid Z_{it} = z)
\;=\;
\underbrace{E\big[p_i(z)\{1 - p_i(z)\}\big]}_{\text{within type}}
\;+\;
\underbrace{\operatorname{Var}\{p_i(z)\}}_{\text{between types}}.
\label{eq:vardecomp}
\end{equation}
A cross section identifies only the sum; two periods identify the summands.

\begin{proposition}[Treatment choice-kernel degeneracy index]
\label{prop:determinacy}Suppose the unconditional version of Assumption~\ref{ass:paneltypes} holds
and $\Pr(Z_{is}=Z_{it}=z)>0$ for some $s\neq t$, with
$0<E[p_i(z)]<1$. Then:
\begin{itemize}
\item[(a)] $\operatorname{Cov}(D_{is}, D_{it} \mid Z_{is} = Z_{it} = z) = \operatorname{Var}\{p_i(z)\}$ and the ratio
\[
\rho(z) \;\equiv\; \frac{\operatorname{Var}\{p_i(z)\}}{E[p_i(z)]\,\{1 - E[p_i(z)]\}} \;\in\; [0,1]
\]
is identified from two periods.
\item[(b)] $\rho(z) = 1$ if and only if $p_i(z) \in \{0,1\}$ almost surely, and $\rho(z) = 0$ if and only if $p_i(z)$ is constant across individuals. The switching probability satisfies
\[
\Pr(D_{is} \neq D_{it} \mid Z_{is} = Z_{it} = z)
\;=\;
2\, E[p_i(z)]\,\{1 - E[p_i(z)]\}\,\{1 - \rho(z)\}=2E\big[p_i(z)\{1 - p_i(z)\}\big].
\]
\item[(c)] Within $M_S$, the stable deterministic boundary $M_D$ holds
at every repeated state $z$ if and only if, for each such state, either
$E[p_i(z)]\in\{0,1\}$ or $\rho(z)=1$. In particular, at states with
$0<E[p_i(z)]<1$, degeneracy is equivalent to $\rho(z)=1$.
Thus, within-person switching at a repeated state refutes $M_D$, although
it does not refute an unrestricted deterministic-path representation.
\end{itemize}
\end{proposition}

\begin{proof}
By Assumption~\ref{ass:paneltypes}(ii) and (iv), conditional on the event $\{Z_{is} = Z_{it} = z\}$ the tuple $(D_{is}, D_{it})$ is, given $G_i$, a pair of independent Bernoulli$(p_i(z))$ draws. Hence $E[D_{is} D_{it} \mid Z_{is} = Z_{it} = z] = E_{G}\{E[D_{is} D_{it} \mid G_i; Z_{is} = Z_{it} = z] \} = E[p_i(z)^2]$ where we use the Law of Iterated Expectations and where the first equality uses iterated expectations and conditional
independence in Assumption~\ref{ass:paneltypes}(ii), and the second uses
the unconditional state-path independence in
Assumption~\ref{ass:paneltypes}(iv), which implies that conditioning on
$Z_{is}=Z_{it}=z$ does not change the distribution of $G_i$.
  Similarly, $E[D_{it} \mid Z_{is} = Z_{it} = z] = E[p_i(z)]$.  This thus allows us to obtain $E\big[p_i(z)\{1 - p_i(z)\}\big]=E[D_{it} \mid Z_{is} = Z_{it} = z]-E[D_{is} D_{it} \mid Z_{is} = Z_{it} = z]$ and, consequently, $\operatorname{Var}\{p_i(z)\}=\operatorname{Var}(D_{it} \mid Z_{it} = z)-E[D_{it} \mid Z_{is} = Z_{it} = z]+E[D_{is} D_{it} \mid Z_{is} = Z_{it} = z]$.  The denominator $E[p_i(z)]\,\{1 - E[p_i(z)]\}$ equals $\operatorname{Var}(D_{it} \mid Z_{it} = z)=E[D_{it}^2 \mid Z_{is} = Z_{it} = z]-E[D_{it} \mid Z_{is} = Z_{it} = z]^2=E[D_{it} \mid Z_{is} = Z_{it} = z]-E[D_{it} \mid Z_{is} = Z_{it} = z]^2=E[p_i(z)]-E[p_i(z)]^2$, where the second equality uses $D_{it}^2=D_{it}$ and the last equality follows from similar steps as above.  The ratio $\rho(z)$ is thus identified, giving (a) and, by direct computation, the switching formula in (b). For the range (i.e., $\rho(z) \in [0,1]$) and the characterizations, note that $p^2 \leq p$ on $[0,1]$, so $\operatorname{Var}(p) = E[p^2] - (E[p])^2 \leq E[p]\{1 - E[p]\}$ with equality if and only if $p^2 = p$ (almost surely(, i.e., $p \in \{0,1\}$ (almost surely); and  $\operatorname{Var}(p) = 0$ if and only if $p$ is constant. Part (c) follows since time-invariant $D_i^{*}(\cdot)$ is the case $p_i(z) = D_i^{*}(z) \in \{0,1\}$.
\end{proof}

The index $\rho(z)$ encodes information that the cross-section along cannot offer. The case $\rho(z)=1$ is the stable deterministic boundary, while $\rho(z)=0$ means that $p_i(z)$ is homogeneous across individuals. If this homogeneity holds at both $z$ and $z'$, responsiveness is constant and the Wald estimand is the ATE. Intermediate values measure the share of observed treatment variance attributed by $\mathcal M_S$ to persistent heterogeneity. In repeated exposures to a fixed state, the time until first treatment is conditionally geometric with hazard $p_i(z)$, so mixing over $p_i(z)$ generates negative duration dependence \citep{blumenetal1955,heckman1981}. The index should therefore be interpreted only under the maintained restrictions that separate heterogeneity from state dependence.

The next result extends this logic from a single state to the contrast $(z', z)$ and delivers the objects that Sections~\ref{sec:familiar}--\ref{sec:monotonicity} treat as latent.

\begin{proposition}[Identification of choice-kernel moments and responsiveness]
\label{prop:choicemoments}
Suppose Assumptions~\ref{ass:paneltypes} and~\ref{ass:pathrich} hold, and write
\[
r_i=r_i(z',z)=p_i(z')-p_i(z).
\]
Then:
\begin{itemize}

\item[(a)] For every $\zeta$ in the support of $\mathbf{Z}_i$ and every
$S\subseteq\{1,\ldots,T\}$,
\[
E\left[
\prod_{t\in S}D_{it}
\,\middle|\,
\mathbf{Z}_i=\zeta
\right]
=
E\left[
\prod_{t\in S}p_i(\zeta_t)
\right].
\]
Consequently, all cross moments
\[
E\left[p_i(z)^a p_i(z')^b\right],
\qquad a+b\leq T,
\]
are identified.

\item[(b)] In particular, $E[r_i^k]$ is identified for every $k\leq T$.
If $E[r_i]\neq0$, then with $T\geq2$ the variance of the normalized IV weight
\[
W_i(z',z)=\frac{r_i}{E[r_i]}
\]
is identified.

\item[(c)] (\emph{Refutable implications.})
Let $m_0=1$ and $m_k=E[r_i^k]$, $k=1,\ldots,T$. Based on these identified
moments, one-sided responsiveness, $r_i\geq0$ almost surely, is feasible if
and only if
\begin{equation}
(m_0,m_1,\ldots,m_T)\in
\mathcal H_T([0,1])
\equiv
\left\{
\left(\int_0^1 r^k\,dF(r)\right)_{k=0}^{T}
:
F\text{ is a probability measure on }[0,1]
\right\}.
\label{eq:hausdorff}
\end{equation}
Thus the truncated Hausdorff moment restrictions provide an exact
finite-$T$ feasibility test for one-sided responsiveness based on the
identified moments of $r_i$.\footnote{The Hausdorff moment problem gives
necessary and sufficient conditions for a sequence
$\{m_0,m_1,m_2,\ldots\}$ to admit a representation
$m_k=\int_0^1 r^k\,dF(r)$ for some probability measure $F$ on the closed
unit interval.}

When $T\geq2$, a particularly simple necessary implication is
\[
m_2\leq m_1.
\]
Hence
\[
E[r_i^2]>E[r_i]
\qquad\Longrightarrow\qquad
\Pr\{r_i(z',z)<0\}>0.
\]
Moreover, within the one-sided model,
\[
m_2=m_1
\qquad\Longleftrightarrow\qquad
r_i\in\{0,1\}\quad\text{almost surely}.
\]

\end{itemize}
\end{proposition}

\begin{proof}
For part (a), iterated expectations give
\[
E\left[
\prod_{t\in S}D_{it}
\,\middle|\,
\mathbf{Z}_i=\zeta
\right]
=
E\left[
E\left\{
\prod_{t\in S}D_{it}
\,\middle|\,
G_i,\mathbf{Z}_i=\zeta
\right\}
\,\middle|\,
\mathbf{Z}_i=\zeta
\right].
\]
By Assumption~\ref{ass:paneltypes}(ii), conditional on
$(G_i,\mathbf{Z}_i)$ the treatment indicators are independent Bernoulli
draws with success probabilities $p_i(Z_{it})$. Therefore,
\[
E\left\{
\prod_{t\in S}D_{it}
\,\middle|\,
G_i,\mathbf{Z}_i=\zeta
\right\}
=
\prod_{t\in S}p_i(\zeta_t).
\]
Assumption~\ref{ass:paneltypes}(iv) implies that conditioning on
$\mathbf{Z}_i=\zeta$ does not change the distribution of $G_i$, and hence
\[
E\left[
\prod_{t\in S}D_{it}
\,\middle|\,
\mathbf{Z}_i=\zeta
\right]
=
E\left[
\prod_{t\in S}p_i(\zeta_t)
\right].
\]
For any nonnegative integers $a$ and $b$ with $a+b\leq T$,
Assumption~\ref{ass:pathrich} allows us to choose a state path and a set
$S$ containing $a$ periods at state $z$ and $b$ periods at state $z'$.
The corresponding observable moment therefore identifies
$E[p_i(z)^a p_i(z')^b]$.

For part (b), the binomial expansion
\[
r_i^k
=
\{p_i(z')-p_i(z)\}^k
=
\sum_{j=0}^k
\binom{k}{j}
(-1)^{k-j}
p_i(z')^j p_i(z)^{k-j}
\]
expresses $E[r_i^k]$ as a linear combination of the cross moments
identified in part (a). When $T\geq2$ and $E[r_i]\neq0$,
\[
\Var\{W_i(z',z)\}
=
\frac{E[r_i^2]-E[r_i]^2}{E[r_i]^2},
\]
so the variance of the normalized weight is identified.

For part (c), based on the moments $(m_0,\ldots,m_T)$, one-sided
responsiveness is feasible exactly when those moments admit a representing
probability measure supported on $[0,1]$. This is precisely the truncated
Hausdorff moment condition in~\eqref{eq:hausdorff}.

For the simple second-moment implication, note that
\[
r(1-r)\geq0
\qquad\text{for every }r\in[0,1].
\]
Thus one-sided responsiveness implies
\[
E[r_i(1-r_i)]
=
m_1-m_2
\geq0,
\]
or equivalently $m_2\leq m_1$. Since $r_i\leq1$ always,
$m_2>m_1$ implies that $r_i(1-r_i)<0$ with positive probability, which
can occur only when $r_i<0$. Hence
\[
E[r_i^2]>E[r_i]
\quad\Longrightarrow\quad
\Pr(r_i<0)>0.
\]

Finally, under one-sided responsiveness,
$r_i(1-r_i)\geq0$ almost surely. Therefore
\[
m_2=m_1
\quad\Longleftrightarrow\quad
E[r_i(1-r_i)]=0
\quad\Longleftrightarrow\quad
r_i(1-r_i)=0\ \text{almost surely},
\]
which is equivalent to $r_i\in\{0,1\}$ almost surely.
\end{proof}

The last characterization concerns the responsiveness distribution, not the
entire treatment-choice kernel. In particular, $r_i=0$ implies only that
$p_i(z')=p_i(z)$; their common value need not be zero or one. For example,
\[
p_i(z)=p_i(z')=0.4
\]
implies $r_i=0$, even though treatment choice remains stochastic at both
states. Thus $m_2=m_1$ within the one-sided model establishes binary
responsiveness, but does not establish that
$p_i(z),p_i(z')\in\{0,1\}$ almost surely. Degeneracy of the treatment-choice
kernel at the repeated states is the stronger restriction characterized in
Proposition~4.

The identified moments can also be used to quantify possible violations of
one-sided responsiveness. Using only $(m_0,\ldots,m_T)$, sharp lower and
upper bounds on the fraction of negatively responsive types are obtained by
minimizing and maximizing
\[
\int_{-1}^{1}\mathbf{1}\{r<0\}\,dF(r)
\]
over all probability distributions $F$ on $[-1,1]$ satisfying
\[
\int_{-1}^{1}r^k\,dF(r)=m_k,
\qquad k=0,\ldots,T.
\]
Thus a finite panel need not identify $\Pr(r_i<0)$, but it can restrict that
fraction through the identified moments. Using all of the cross moments of
$(p_i(z),p_i(z'))$ identified in Proposition~\ref{prop:choicemoments}(a)
can sharpen these restrictions by imposing the corresponding moment problem
directly on
\[
\{(u,v)\in[0,1]^2:u\leq v\}.
\]

\begin{corollary}[Long-panel identification of responsiveness]
\label{cor:longpanelchoice}
Suppose Assumptions~\ref{ass:paneltypes} and~\ref{ass:pathrich} hold. If each
of $z$ and $z'$ is visited infinitely often per individual as $T\to\infty$,
then the population long-panel law identifies the joint distribution $F_P$
of
\[
P_i=\bigl(p_i(z),p_i(z')\bigr),
\]
and hence the distribution of $r_i(z',z)$. If $E[r_i(z',z)]\neq0$, it also
identifies the distribution of $W_i(z',z)$.
\end{corollary}

\begin{proof}
Conditional on $(G_i,\mathbf Z_i)$, Assumption~\ref{ass:paneltypes}(ii)
implies that treatment choices observed on the subsequence of periods for
which $Z_{it}=z$ are independent Bernoulli draws with common success
probability $p_i(z)$. If state $z$ is visited infinitely often, the
corresponding within-person treatment frequency therefore converges almost
surely to $p_i(z)$. The analogous frequency for periods with $Z_{it}=z'$
converges almost surely to $p_i(z')$. Hence the long-panel law identifies
the joint distribution of
\[
P_i=\bigl(p_i(z),p_i(z')\bigr).
\]
Since
\[
r_i(z',z)=p_i(z')-p_i(z)
\]
is a known function of $P_i$, its distribution is identified. When
$E[r_i(z',z)]\neq0$, the same is true of
\[
W_i(z',z)=\frac{r_i(z',z)}{E[r_i(z',z)]}.
\]
\end{proof}

Proposition~\ref{prop:choicemoments}(c) gives the monotonicity discussion an empirical handle that the cross section cannot. One-sided responsiveness---the condition under which the Wald ratio is a convex average by Proposition~\ref{prop:onesided}---requires all of the moment restrictions in \eqref{eq:hausdorff}. Their failure is consistent with the information-refinement channel in Section~\ref{sec:information}, which can move treatment probabilities in both directions. Proposition~\ref{prop:choicemoments}(b) also makes the interpretation quantitative at $T=2$: the identified dispersion of $W_i$ measures the scope for the Wald estimand to differ from the ATE for a given amount of treatment-effect heterogeneity, complementing the sensitivity analysis of \citet{small2017}.

\subsection{Outcome Data and Treatment-Effect Heterogeneity}
\label{subsec:panel-outcomes}

We now add outcomes and maintain the exclusion restriction in
Assumption~\ref{exclusion}, so that $\mu_i(d,z)=\mu_i(d)$ and
$\Delta_i=\mu_i(1)-\mu_i(0)$. Repeated observations then identify moments
linking individual mean outcomes to individual treatment probabilities. With a
finite panel, these are polynomial co-moments. In the long-panel limit, they
identify how mean treatment effects vary with
$P_i=(p_i(z),p_i(z'))$, except for types who never reveal one of the two
outcome arms. The basic identity is as follows. For any period $t$, disjoint
$A,B\subseteq\{1,\ldots,T\}\setminus\{t\}$, and $\zeta$ in the support of
$\mathbf Z_i$,
\begin{equation}
E\Big[Y_{it}D_{it}\prod_{s\in A}D_{is}
                  \prod_{u\in B}(1-D_{iu})
  \,\Big|\,\mathbf Z_i=\zeta\Big]
=
E\Big[\mu_i(1)p_i(\zeta_t)
       \prod_{s\in A}p_i(\zeta_s)
       \prod_{u\in B}\{1-p_i(\zeta_u)\}\Big].
\label{eq:outcomeidentity}
\end{equation}
There is an analogous identity for the untreated outcome, obtained by
replacing the leading $D_{it}$ by $1-D_{it}$, $\mu_i(1)$ by $\mu_i(0)$,
and $p_i(\zeta_t)$ by $1-p_i(\zeta_t)$. Equation~\eqref{eq:outcomeidentity}
follows from Assumption~\ref{ass:paneltypes}(ii)--(iv) by iterated
expectations. Each treatment indicator from another period contributes the
corresponding factor $p_i(\zeta_s)$, while each non-treatment indicator
contributes $1-p_i(\zeta_u)$.

Writing $u=p_i(z)$ and $v=p_i(z')$, the panel therefore identifies
$E[\mu_i(1)g(u,v)]$ for polynomials $g$ of degree at most $T$ satisfying
$g(0,0)=0$, and $E[\mu_i(0)g(u,v)]$ for polynomials of degree at most $T$
satisfying $g(1,1)=0$. These restrictions have a direct interpretation. Types
with $(u,v)=(0,0)$ never reveal their treated outcome mean, while types with
$(u,v)=(1,1)$ never reveal their untreated outcome mean. With finite $T$, the
panel identifies only the corresponding polynomial co-moments. In the
long-panel limit, the two corners are the only points at which the relevant
conditional mean remains unidentified.

\begin{proposition}[Treatment effects conditional on treatment probabilities]
\label{prop:paneloutcomes}
Suppose Assumptions~\ref{exclusion}, \ref{ass:paneltypes}
and~\ref{ass:pathrich} hold, $E[|\mu_i(d)|]<\infty$ for $d=0,1$, and let
$P_i=(u_i,v_i)=(p_i(z),p_i(z'))$ and $r_i=v_i-u_i$.
\begin{itemize}
\item[(a)] $E[\Delta_i g(P_i)]$ is identified for every polynomial $g$ of
degree at most $T$ satisfying $g(0,0)=g(1,1)=0$. In particular,
$E[\Delta_i r_i^{\,k}]$ is identified for $1\leq k\leq T$. The case $k=1$
is the reduced form of Section~\ref{sec:familiar}; the higher-order
co-moments are additional information supplied by the panel.
\item[(b)] Suppose each of $z$ and $z'$ is visited infinitely often per
individual as $T\to\infty$. Define
\[
m_d(P)\equiv E[\mu_i(d)\mid P_i=P],
\qquad
\gamma(P)\equiv m_1(P)-m_0(P).
\]
The population long-panel law identifies $F_P$, $m_1(P)$ for
$F_P$-almost every $P\neq(0,0)$, $m_0(P)$ for $F_P$-almost every
$P\neq(1,1)$, and therefore $\gamma(P)$ for $F_P$-almost every
$P\notin\{(0,0),(1,1)\}$.
\item[(c)] The treatment-effect schedule conditional on responsiveness,
\[
\beta(r)\equiv E[\Delta_i\mid r_i=r],
\]
is identified for $F_r$-almost every $r\neq0$. The restriction $r\neq0$
excludes both unidentified corners. The more detailed schedule $\gamma(P)$ is
also identified for types with $r=0$ whenever $P$ is not one of the two
corners.
\end{itemize}
\end{proposition}

\begin{proof}[Proof sketch]
For part (a), any polynomial $g$ satisfying $g(0,0)=0$ is a linear
combination of monomials $u^av^b$ of positive degree. Each such monomial of
degree at most $T$ is obtained from \eqref{eq:outcomeidentity} by placing the
outcome $Y_{it}$ in one of the corresponding treatment periods. Rewriting the
polynomial in powers of $(1-u,1-v)$ gives the analogous result for $\mu_i(0)$
when $g(1,1)=0$. Both restrictions are therefore required to identify
$E[\Delta_i g(P_i)]$. Since $r=v-u$ vanishes at both corners, $r^k$ satisfies
these restrictions.

For part (b), as $T\to\infty$, the identified co-moments determine the finite
signed measure $m_1(P)dF_P(P)$ away from $(0,0)$ and the measure
$m_0(P)dF_P(P)$ away from $(1,1)$. This follows because polynomials are dense
in the continuous functions on the compact set $[0,1]^2$, and the relevant
polynomials vanish only at the corner where the corresponding outcome arm is
never observed. Corollary~\ref{cor:longpanelchoice}  identifies $F_P$.
The Radon--Nikodym derivatives therefore recover $m_1$ and $m_0$ on the
stated sets, and hence $\gamma$. 

Part (c) follows by integrating $\gamma(P)$
with respect to the identified conditional distribution of $P$ given
$v-u=r$. If $r\neq0$, neither unidentified corner can occur.
\end{proof}

The Wald estimand $\theta(z',z)$ is already identified from the cross section
when $E[r_i(z',z)]\neq0$. The additional long-panel content is that the
weights are also identified.

\begin{corollary}[IV weights in long panels]
\label{cor:paneliv}
Under the conditions of Proposition~\ref{prop:paneloutcomes}, $A^+$ and $A^-$
in Section~\ref{sec:monotonicity} are identified. The corresponding averages
$\theta^+$ and $\theta^-$ are identified whenever their respective
denominators are nonzero. The same argument applies to a finite collection of
states $\mathcal Z_0$ that are each visited infinitely often. It identifies
the treatment-effect schedule conditional on
$P_i^{\mathcal Z_0}=(p_i(z):z\in\mathcal Z_0)$ away from the all-zero and
all-one corners. Under conditional state independence, the argument is applied
jointly with the observed covariates $X_i$. Consequently, the weights of the 2SLS estimand in Section~\ref{subsec:tsls} is identified
whenever $\kappa_i^{2SLS}$ is an identified function of
$(P_i^{\mathcal Z_0},X_i)$ and $E[\kappa_i^{2SLS}]\neq0$.
\end{corollary}

The last statement requires the full vector of treatment probabilities, not
only $r_i(z',z)$. For a discrete multivalued instrument,
$\kappa_i^{2SLS}$ is a linear combination of the components of
$P_i^{\mathcal Z_0}$ whose coefficients sum to zero. It therefore vanishes at
the all-zero and all-one corners, exactly where one of the two outcome means is
not identified. The conditional treatment-effect schedule is therefore
sufficient to recover the treatment-effect average represented by the 2SLS
estimand.

This result gives empirical content to the decomposition in
Section~\ref{sec:monotonicity}. When responsiveness changes sign,
Corollary~\ref{cor:paneliv} identifies the positive and negative components
separately, so the analyst can report the two convex averages rather than only
their signed ratio. The same schedule also shows how different instruments or
2SLS specifications produce different treatment-effect averages through their
different treatment-probability weights.

Repeated outcomes can also sharpen what is learned about the ATE, even when
individual treatment effects are not fully identified.

\begin{corollary}[ATE bounds from repeated outcomes]
\label{cor:panelate}
Suppose Assumptions~\ref{exclusion} and~\ref{ass:paneltypes} hold and
$Y_{it}(d,z)\in[\underline y,\overline y]$ almost surely for every $d$ and
$z$. Let
\[
a_{iT}=\prod_{t=1}^{T}\{1-p_i(Z_{it})\},
\qquad
b_{iT}=\prod_{t=1}^{T}p_i(Z_{it}),
\]
and define
\begin{align}
L_T={}&E[\mu_i(1)(1-a_{iT})]+\underline yE[a_{iT}]
      -E[\mu_i(0)(1-b_{iT})]-\overline yE[b_{iT}], \nonumber\\
U_T={}&E[\mu_i(1)(1-a_{iT})]+\overline yE[a_{iT}]
      -E[\mu_i(0)(1-b_{iT})]-\underline yE[b_{iT}].
\label{eq:panelatebounds}
\end{align}
Every term in \eqref{eq:panelatebounds} is identified and
\begin{equation}
ATE\in[L_T,U_T],
\qquad
U_T-L_T=(\overline y-\underline y)\{E[a_{iT}]+E[b_{iT}]\}.
\label{eq:panelatewidth}
\end{equation}
If $a_{iT}\to0$ and $b_{iT}\to0$ almost surely, then the width in
\eqref{eq:panelatewidth} converges to zero. A sufficient condition is that
some fixed state $z$ is visited infinitely often and $0<p_i(z)<1$ almost
surely. This sufficient condition requires no within-person variation in the
state.
\end{corollary}

\begin{proof}
Let \(A_t=\{D_{it}=1\}\). Conditional on \(G_i\) and the realized state path
\(\mathbf Z_i\), Assumption 4(ii) implies
\[
a_{iT}
=
\prod_{t=1}^T \{1-p_i(Z_{it})\}
=
\Pr\left(
\bigcap_{t=1}^T A_t^c
\;\middle|\;
G_i,\mathbf Z_i
\right).
\]
Hence
\[
1-a_{iT}
=
\Pr\left(
\bigcup_{t=1}^T A_t
\;\middle|\;
G_i,\mathbf Z_i
\right),
\]
the conditional probability that treatment is observed at least once.
Applying the inclusion--exclusion formula to this union gives
\[
\begin{aligned}
1-a_{iT}
={}&
\sum_t p_i(Z_{it})
-\sum_{s<t} p_i(Z_{is})p_i(Z_{it}) \\
&+\sum_{r<s<t}
p_i(Z_{ir})p_i(Z_{is})p_i(Z_{it})
-\cdots .
\end{aligned}
\]
Therefore \(E[\mu_i(1)(1-a_{iT})]\) is a finite linear combination of
outcome--choice moments of the form in equation~(13).

The unobserved remainder \(E[\mu_i(1)a_{iT}]\) lies between
\[
\underline y\,E[a_{iT}]
\quad\text{and}\quad
\overline y\,E[a_{iT}].
\]
Moreover,
\[
E[a_{iT}]
=
\Pr(D_{i1}=\cdots=D_{iT}=0),
\]
after integrating over the state design. The argument for the untreated
outcome and \(b_{iT}\) is analogous. The convergence result follows by
dominated convergence.
\end{proof}

These are valid finite-$T$ bounds; they need not be sharp after all identified
panel moments are imposed. Their role is to provide an observable measure of
the remaining missing-outcome problem. At a single fixed state, interior
stochastic choice supplies within-person treatment overlap and the width can
vanish with repeated observations. Stable deterministic choice at that same
state cannot provide this overlap. With several states, however, deterministic
switchers may reveal both outcome arms. In that case, the limiting width is
positive only if there is positive probability of types who are never treated
or always treated over the set of visited states. The comparison is therefore
between alternative sources of within-person treatment variation, not between
stochastic and deterministic potential outcomes in general.

\subsection{Discussion: What the Panel Assumptions Buy, and What They Cost}
\label{subsec:panel-discussion}

Assumption~\ref{ass:paneltypes} does real work. Serial correlation in residual choice draws, structural state dependence and time-varying types all generate persistence that Proposition~\ref{prop:determinacy} would attribute to stable heterogeneity. This is the familiar problem of separating heterogeneity from state dependence \citep{heckman1981,chamberlain1984,chamberlain1985}. The maintained model has refutable implications: conditional on the multiset of states visited, the choice-path distribution is invariant to permutations of periods that leave states unchanged. Rejection of this exchangeability restriction rejects the maintained panel model, although it does not by itself distinguish state dependence, serial correlation and changing types.

Dynamic binary-choice models provide useful templates for relaxing the assumption that current choice depends only on the current state and the stable type. Finite mixtures of Markov processes, unrestricted transition heterogeneity and dynamic bounds are studied by \citet{kasaharashimotsu2009,browningcarro2014,honoretamer2006,torgovitsky2019}. Their results do not transfer mechanically to the present outcome--choice co-moments: a dynamic extension would require additional restrictions and a separate identification analysis. Elicited choice probabilities provide a complementary source of information. When stated probabilities can be interpreted and validated as measurements of $p_i(z)$, they can reduce the amount of repetition needed; a short panel can in turn be used to assess their calibration.

The results in this section concern conditional means. Assumption~\ref{ass:paneltypes}(iii) does not impose distributional independence of outcome shocks from the treatment path or conditional independence of outcome draws across periods. Even under those stronger restrictions, repeated outcome products first identify $p_i(z)$-weighted second moments of $\mu_i(d)$, not their unweighted counterparts. The present results therefore do not separate the between-type variance of mean outcomes from the within-type variance of the outcome kernels $Q_i$. Identifying those features would require stronger distributional assumptions and additional support or long-panel arguments.

Relative to work on distributions of heterogeneous effects in nonlinear panels \citep{arellanobonhomme2012,chernozhukovetal2013}, the distinctive object here is the joint dependence between treatment effects and the individual treatment-probability movements induced by an instrument or policy. That dependence is the weighting operator behind the IV estimands studied in Sections~\ref{sec:familiar}--\ref{sec:monotonicity}. Relative to the cross-sectional stochastic-monotonicity framework of \citet{small2017}, the stable type $p_i(z)$ imposes restrictions across repeated choices. When those restrictions are not rejected, the weights cease to be only an interpretation and become identified objects.

\section{Conclusion }
\label{sec:conclusion}

This paper reinterprets familiar causal estimands under stochastic potential outcomes and stochastic treatment choice. Each individual carries a stable response type consisting of a treatment choice kernel and an outcome kernel. Under state independence, an observed outcome contrast identifies the causal effect of the assigned state. Adding exclusion writes that state effect as a contrast in treatment effects weighted by how much the state changes each individual's probability of treatment.

The standard compliance taxonomy is a limiting case. Uniform responsiveness gives the ATE, one-sided responsiveness gives a convex average, and sign-changing responsiveness gives a signed contrast. The same logic applies to 2SLS, where the weighting operator depends on the first-stage specification as well as on instrument validity. The central empirical question is therefore whose treatment probabilities move, by how much, and in which direction.

Repeated choices give this interpretation additional content under explicit panel restrictions. Within the stable, conditionally independent kernel model, short panels identify moments of individual treatment probabilities and make one-sided responsiveness testable. Long panels identify the joint dependence between mean treatment effects and the vector of treatment probabilities across states, thereby identifying the weights behind binary IV and multivalued 2SLS estimands. The stable deterministic-response model is the degenerate boundary of this maintained panel model, rather than the only deterministic representation of a choice path.

Elicited choice probabilities offer a complementary route to the same objects when they can be interpreted as measurements of $p_i(z)$. An important next step is to allow structural state dependence, serially correlated choice shocks and changing information sets, and to determine which parts of the joint effect--propensity schedule remain identified under those more general dynamics.

\newpage
\appendix

\section{Static Observational Equivalence of Ex Ante and Ex Post Representations}
\label{app:equivalence}

This appendix records the sense in which an ex ante stochastic representation and an ex post deterministic representation are (mean) observationally equivalent in a static environment.  In a traditional discrete choice setting, this relates to the observation that ``the assumption that a single subject will draw independent utility functions in repeated choice settings and then proceed to maximize them is equivalent to a model in which the experimenter draws individuals randomly from a population with differing, but fixed, utility functions'' \citep{mcfadden1976}.

Consider a binary state $Z_i\in\{0,1\}$, binary treatment $D_i\in\{0,1\}$, and outcome $Y_i\in\cY\subseteq\R$. We consider the exclusionrestriction version of the model for simplicity where $Q_i(.|d,z) = Q_i(.|d)$. In the (ex ante) stochastic representation ($\mathcal{M}_S$), there exists a latent type
\[
        G_i=(p_i(0),p_i(1),Q_i(\cdot\mid 0),Q_i(\cdot\mid 1))
\]
such that $Z_i\Perp G_i$, treatment is Bernoulli with probability $p_i(z)$ under state $z$, and potential outcomes are drawn from the type-specific outcome kernels. In the (ex post) deterministic representation ($\mathcal{M}_D$), there exist jointly defined variables
\[
        D_i^*(0),D_i^*(1)\in\{0,1\},
        \qquad
        Y_i^*(0),Y_i^*(1)\in\cY,
\]
with $Z_i\Perp(D_i^*(0),D_i^*(1),Y_i^*(0),Y_i^*(1))$, $D_i=D_i^*(Z_i)$, and $Y_i=Y_i^*(D_i)$.

\begin{proposition}[Static observational equivalence]\label{prop:equivalence}The (ex ante) stochastic representation and the (ex post) deterministic representation generate the same set of observable conditional mean functions
\[
z\mapsto
\left(E[D_i\mid Z_i=z],\,E[Y_i\mid Z_i=z]\right).
\]Consequently, they generate the same reduced forms, first stages, and Wald estimands in the static binary-state environment described above.
\end{proposition}

\begin{proof}[Proof sketch]
Start with an ex ante stochastic representation. Enlarge the probability space with a sufficiently rich random vector $(U_{i0},U_{i1},V_{i0},V_{i1})$ where $U_{iz},z\in\{0,1\}$ and $V_{id},d\in\{0,1\}$ are uniformly distributed and mutually independent of $Z_i$ conditional on $G_i$. Define
\[
        D_i^*(z)=\1\{U_{iz}\leq p_i(z)\}
\]
and define $Y_i^*(d)$ by applying the quantile function of $Q_i(\cdot\mid d)$ to $V_{id}$. Conditional on $G_i$, the law of $(Y_i^*(d),D_i^*(z))$ matches the stochastic kernels.

Then
 $$
 \{D_i^*(0),D_i^*(1),Y_i^*(0),Y_i^*(1)\}\mid G_i
$$ leads to \[
E[D_i^*(z)\mid G_i]=p_i(z),
\]and, using independence of the constructed uniforms,
\[
\begin{aligned}
E[Y_i^*(D_i^*(z))\mid G_i]
&=\sum_{d\in\{0,1\}}
P(D_i^*(z)=d\mid G_i)E[Y_i^*(d)\mid G_i]\\
&=\sum_{d\in\{0,1\}}p_i(d\mid z)\mu_i(d)\\
&=m_i(z).
\end{aligned}
\]Under Assumption 1, these are exactly the corresponding conditional means in the original stochastic representation. State independence then implies that both representations generate identical \(E[D_i\mid Z_i=z]\) and \(E[Y_i\mid Z_i=z]\).
Conversely, start with an ex post deterministic representation and set $p_i(z)=D_i^*(z)$ and $Q_i(\cdot\mid d)$ equal to the point mass at $Y_i^*(d)$. This is a degenerate ex ante stochastic representation that produces the same observable law and hence mean functions.
\end{proof}

The equivalence is a statement about observable conditional  means in a static cross-section. It does not imply equality of the latent objects, and it need not survive additional restrictions involving dynamics, repeated choices, panel dependence, restrictions on shocks, or assumptions about how treatment effects vary with transient choice states. The stochastic representation is therefore best understood as an ex ante structural parameterization: it treats $z\mapsto p_i(z)$ as the stable object and interprets IV weights as changes in this map.
\newpage
\bibliographystyle{apalike}
\bibliography{biblio}
\end{document}